\pdfoutput=1
\documentclass[journal,letterpaper,twoside]{IEEEtran}
\usepackage{multirow}
\usepackage{bm}
\usepackage{graphicx,cite}
\usepackage[T1]{fontenc}

\usepackage{comment}
\usepackage{diagbox}
\usepackage{float}
\usepackage{cite}
\usepackage{amsmath, amssymb,amsthm}
\usepackage[ansinew]{inputenc}

\newtheorem{definition}{Definition}
\newtheorem{theorem}{Theorem}
\newtheorem{proposition}{Proposition}
\newtheorem{lemma}{Lemma}

\usepackage{setspace}
\usepackage{graphics}

\usepackage{algorithm}
\usepackage{algorithmic}
\usepackage{array}
\usepackage{longtable,enumitem}

\usepackage{listings}
\usepackage{color}
\usepackage{csquotes}

\usepackage[caption=false,font=footnotesize]{subfig}
\usepackage[colorlinks=true,linkcolor=black,citecolor=black,urlcolor=black]{hyperref}
\usepackage{cleveref}
\crefname{section}{\S}{\S\S}
\Crefname{section}{\S}{\S\S}
\crefformat{section}{\S#2#1#3}
\hypersetup{colorlinks={true},linkcolor={blue},citecolor=blue}

\newcommand{\telu}{$TE_{LU}$}
\newcommand{\temf}{$TE_{MF}$}

\begin{document}


% Copyright
% \setcopyright{acmcopyright}
%\setcopyright{acmlicensed}
%\setcopyright{rightsretained}
%\setcopyright{usgov}
%\setcopyright{usgovmixed}
%\setcopyright{cagov}
%\setcopyright{cagovmixed}

% DOI
% \doi{10.475/123_4}

% ISBN
% \isbn{123-4567-24-567/08/06}

%Conference
% \conferenceinfo{SIGMETRICS '17}{June 5--9, 2016, Champaign-Urbana, IL, USA}

% \acmPrice{\$15.00}

%
% --- Author Metadata here ---
% \conferenceinfo{WOODSTOCK}{'97 El Paso, Texas USA}
%\CopyrightYear{2007} % Allows default copyright year (20XX) to be over-ridden - IF NEED BE.
%\crdata{0-12345-67-8/90/01}  % Allows default copyright data (0-89791-88-6/97/05) to be over-ridden - IF NEED BE.
% --- End of Author Metadata ---

\title{Centrality-based Middlepoint Selection for\\
Traffic Engineering with Segment Routing% in SDN based WANs
}
% \title{MidwayTE: Traffic Engineering with Segment Routing in SDN based WANs }

\author{George~Trimponias,
        Yan~Xiao,
        Hong~Xu,
        Xiaorui~Wu,
        and~Yanhui~Geng% <-this % stops a space
% \thanks{The work was supported in part by contract research between City University of Hong Kong and Huawei. The corresponding author is Hong Xu.}
\thanks{G. Trimponias and Y. Geng are with Huawei Noah's Ark Lab (email: \href{mailto:g.trimponias@huawei.com}{g.trimponias@huawei.com}, \href{mailto:geng.yanhui@huawei.com}{geng.yanhui@huawei.com}). }% <-this % stops a space
\thanks{Y. Xiao, H. Xu and X. Wu are with Department of Computer Science, City University of Hong Kong, Hong Kong, China (email: \href{mailto:yanxiao6-c@my.cityu.edu.hk}{yanxiao6-c@my.cityu.edu.hk}, \href{mailto:henry.xu@cityu.edu.hk}{henry.xu@cityu.edu.hk}, \href{mailto:xiaoruiwu3-c@my.cityu.edu.hk}{xiaoruiwu3-c@my.cityu.edu.hk}).}% <-this % stops a space
}

% \markboth{IEEE/ACM TRANSACTIONS ON NETWORKING}%
% {Trimponias et al.: On Traffic Engineering with Segment Routing in SDN based WANs}

\maketitle

\begin{abstract}

Segment routing is an emerging technology to simplify traffic engineering
implementation in WANs. It expresses an end-to-end logical path as a sequence
of segments through a set of middlepoints. Traffic along each segment is
routed along shortest paths. In this paper, we study practical traffic
engineering (TE) with segment routing in SDN based WANs. We consider two
common types of TE, and show that the TE problem can be solved in weakly
polynomial time when the number of middlepoints is fixed and not part of the
input. However, the corresponding linear program is of large scale and
computationally expensive. For this reason, existing methods that work by
taking each node as a candidate middlepoint are inefficient. Motivated by
this, we propose to select just a few important nodes as middlepoints for all
traffic. We use node centrality concepts from graph theory, notably group
shortest path centrality, for middlepoint selection. Our performance
evaluation using realistic topologies and traffic traces shows that a small
percentage of the most central nodes can achieve good results with orders of
magnitude lower runtime.

\end{abstract}

\begin{IEEEkeywords}
Segment Routing, Traffic Engineering, Graph Centrality, Software Defined Networking
\end{IEEEkeywords}

% \keywords{Segment Routing, Traffic Engineering, Software Defined Networking, Source Routing, Graph Centrality}

%!TEX root = main.tex
\section{Introduction}
\label{sec:intro}

% TE is an important task and is done mostly in SDN now
Traffic engineering (TE) is an important task for network operators to improve network efficiency and application performance. TE is commonly exercised in a wide range of networks, from carrier networks \cite{HVSB15,FT00} to data center backbones \cite{JKMO13,HKMZ13}. Increasingly, TE is implemented using SDN (Software Defined Networking) given its flexibility. Notable examples include Google's B4 \cite{JKMO13} and Microsoft's SWAN \cite{HKMZ13}. 
% Usually some tunneling protocol is used: the controller
% establishes multiple tunnels (i.e. network paths) between
% an ingress-egress switch pair, and configures splitting weights at
% the ingress switch. 

% The ingress switch then uses hashing based
% multipath forwarding such as ECMP to send flows [6, 7, 9]. 

% TE involves high cost in terms of the routing table entries for SDN switches. 
Implementing TE in the data plane requires a large number of flow table entries on switches. This is because each switch on the path needs to have an entry for a demand, i.e. ingress-egress switch pair, to forward its traffic to the next hop, and for a large-scale network there can be many demands. Commodity switches, on the other hand, have very limited capacity for flow entries (usually 1-2 thousands of entries \cite{CMTY11,KPK15}) due to the expensive TCAM (Ternary Content Aware Memory) hardware needed \cite{CMTY11,CLNR14,MPEK16}. The use of wildcarding could reduce the number of flow entries, but it is often undesirable as it reduces the ability to implement demand-level policies and monitoring. 
% Moreover, it is very slow to update an entry on an Openflow switch ($\sim$10ms) \cite{CMTY11}, making it even more difficult and error-prone to implement hundreds or thousands of TE rule changes every 5 or 10 minutes. 
Therefore it has become a major challenge to practically implement TE on commodity SDN switches. 

% existing approaches to reduce routing table entires; move this to related work
% segment routing as a good candidate developed earlier

Segment routing \cite{segment1, segment2, segment3} is a recently proposed routing architecture to tackle this challenge. 
Its key idea is to perform routing based on a sequence of logical segments formed by some {\em middlepoints} between the ingress and egress nodes. A segment is the logical pipe between two middlepoints that may include multiple physical paths spanning multiple hops, and ECMP is used to load balance traffic among these paths.  
Now with segment routing instead of end-to-end paths, intermediate switches only need to know how to reach middlepoints in order to forward packets. They no longer need to maintain per-demand routing information which scales quadratically with the number of nodes.
Thus segment routing has the potential to greatly reduce the overhead and cost of TE \cite{bhatia2015optimized,HVSB15}. 

% challenge: how to pick the middlepoints and how to use them for TE?
Segment routing has been explored with TE in some existing work. For example Bhatia et al. \cite{bhatia2015optimized} apply 2-segment routing to TE, where any logical path contains only one middlepoint and thus two segments. Hartert et al. \cite{HVSB15} propose some heuristics to solve various TE problems with segment routing. 
% Aubry et al. \cite{aubry2016scmon} applies segment routing to network monitoring and measurement. 
There lacks a thorough exploration and understanding of applying segment routing to TE, particularly the hardness of the resulting TE problem, and the development of practical TE algorithms with segment routing.  
% To utilize segment routing in TE, two fundamental questions must be addressed: first, which middlepoint(s) should be used to route the demands, and second, how the routing should be done across these segments? Clearly performance is maximized when the two are jointly considered, which is impossible to solve for however due to the exponentially many possibilities of combinations of segments as paths. Several recent works have started to tackle this problem \cite{bhatia2015optimized,HVSB15}. The general approach is to consider, for each demand, the joint problem of middlepoint selection and TE optimization, and iteratively optimize the objective. 
% This approach still suffers from high complexity largely due to the combinatorial nature of per-demand based middlepoint selection, because every node is a potential candidate of middlepoint. Indeed Hartert et al. prove that middlepoint selection is NP-hard even if there are only link capacity constraints \cite{HVSB15}. Thus greedy heuristics are usually adopted to derive a solution.  

In this paper, we study practical TE with segment routing in SDN based WANs. We first focus on some theoretical aspects of TE with segment routing. We consider two common types of TE: $TE_{MF}$ that maximizes total throughput based on multi-commodity flow, and $TE_{LU}$ that minimizes the maximum link utilization. $TE_{MF}$ is mostly for data center backbone WANs \cite{JKMO13,HKMZ13}, and $TE_{LU}$ mostly for carrier networks \cite{HVSB15,FT00}. 
We analyze study segment routing that uses only shortest paths between two segments in \cref{sec:opt}, and prove that both TE problems can be solved in (weakly) polynomial time as an LP when the number of middlepoints per path is fixed and not part of the input. Our results thus provide a theoretical foundation for existing work that focuses on shortest path based segment routing \cite{bhatia2015optimized,HVSB15}. 

We next focus on the practical problem of how to choose a small but representative set of middlepoints in order to solve TE with shortest path based segment routing (\cref{sec:2pass}). 
Existing approaches \cite{bhatia2015optimized} assume that for each demand, every node in the network is potentially a middlepoint candidate, and formulate it as part of the TE problem. This causes the TE to be of a very large scale, which makes it computationally expensive to solve for practical purposes. 
As we show in \cref{sec:vsTE}, it cannot be solved by the {ECOS solver} \cite{D13} after three hours on a medium topology with 100 nodes and 1500 commodity flows, while in practice TE often needs to be re-computed at the granularity of 10 minutes \cite{JKMO13,HVSB15,HKMZ13,LKMZ14}. 
% Some other work relies on heuristics to do middlepoint selection \cite{HVSB15} 

We thus propose to apply the centrality concept from graph theory and network analysis  \cite{Newman2010} to select a few middlepoints to route all traffic in the network. 
Centrality was first developed in social network analysis \cite{F1977,b87} to determine the most influential nodes in a social graph. 
In the context of routing, centrality can be naturally viewed in terms of the node importance when routing the demands along the admissible paths. 
We explore several centrality definitions based on the network topology only, such as shortest-path, group shortest-path, and degree centralities, and apply them to middlepoint selection in networks. We also introduce weighted variants that additionally take into account the link capacities.

We conduct comprehensive performance evaluation of centrality based middlepoint selection using real topologies and traffic traces. Our results demonstrate that only a small percentage of around 2.5\%--7\% of the most central nodes can achieve good TE performance with orders of magnitude lower runtime. Using centrality based middlepoint selection methods, one can solve TE problems with up to 3000 flows on a 161-node topology in less than 3 minutes. We also observe that group shortest-path consistently outperforms other centralities for middlepoint selection, and may be used as the sole solution in practice for simplicity.  

\section{A Primer on Segment Routing }
\label{sec:bg}
% \subsection{Traffic Engineering}
We start by introducing segment routing and the benefit of applying it to TE. We next explain related work on segment routing.

% \subsection{Segment routing}
% \label{sec:sr}

Segment routing \cite{segment1, segment2, segment3} is a recently proposed architecture based on source routing that facilitates packet forwarding via a series of segments. It can be directly applied to MPLS and IPv6.
The key idea is that the ingress switch can break up the end-to-end logical path into segments, and specify this logical path as a series of {\em middlepoints} to traverse.
Figure~\ref{SR} illustrates an example of segment routing. The ingress switch S embeds a stack of segment labels (MPLS labels for example) into the packet header to specify the entire path. Note here each label just represents a middlepoint in the network, i.e. we consider node segments here \cite{segment1, segment2, segment3}. The top label is the active label that instructs packet forwarding. Then the packet is sent to the next label M$_1$ along the shortest path(s). ECMP is used if there are multiple shortest paths. When the packet reaches M$_1$, the top label M${_1}$ is popped and the packet is routed to the next label M$_2$. Finally, all the labels are popped and the packet arrives at the egress switch D.

% Generally, a segment is an MPLS label
% where has three operations: push, pop and continue.  The segment has different types and in this paper we introduce the node segment as node segments can be used to describe paths in network topology.   Node segment contains the unique Node-SID, which identify the next router and  a list of node segments defines the path from ingress to egress. Note that between two segments, the segment routing  explores the Equal-Cost Multi-Paths (ECMP).

\begin{figure}[H]
\begin{center}
\includegraphics [scale = 0.6]{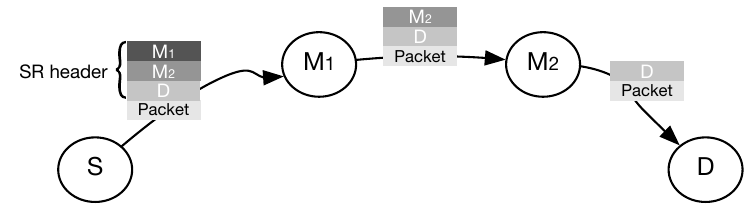}
% \vspace{-4mm}
\caption{A segment routing example from S to D through middlepoints M$_1$ and M$_2$.}
\label{SR}
\end{center}
\vspace{-4mm}
\end{figure}

One key advantage of segment routing is that it can greatly reduce routing cost in terms of number of flow table entries required. To see this, consider the next example shown in Figure~\ref{fig:showcase}. Three {\em demands}\footnote{When it is clear from the context, we use the terms {\em commodities}, {\em demands}, and {\em flows} interchangeably.}, which refer to the aggregated flows between a unique ingress-egress switch pair, are routed through three paths P1, P2, and P3 to their respective destinations. With tunnel-based forwarding in SDN \cite{JKMO13,HKMZ13}, each intermediate switch needs to store flow entries for each demand, and in total 12 entries are needed as shown in Table~\ref{table:rules}. Now if segment routing is applied with node E as the middlepoint, the three paths can be represented using just two labels each as in Figure~\ref{fig:showcase}, and switches C and D only need to have one entry in order to forward to the middlepoint E. The total number of entries is reduced to only 8, with 33.3\% saving.

\begin{figure}[htp]
\centering
\includegraphics[width=0.9\linewidth]{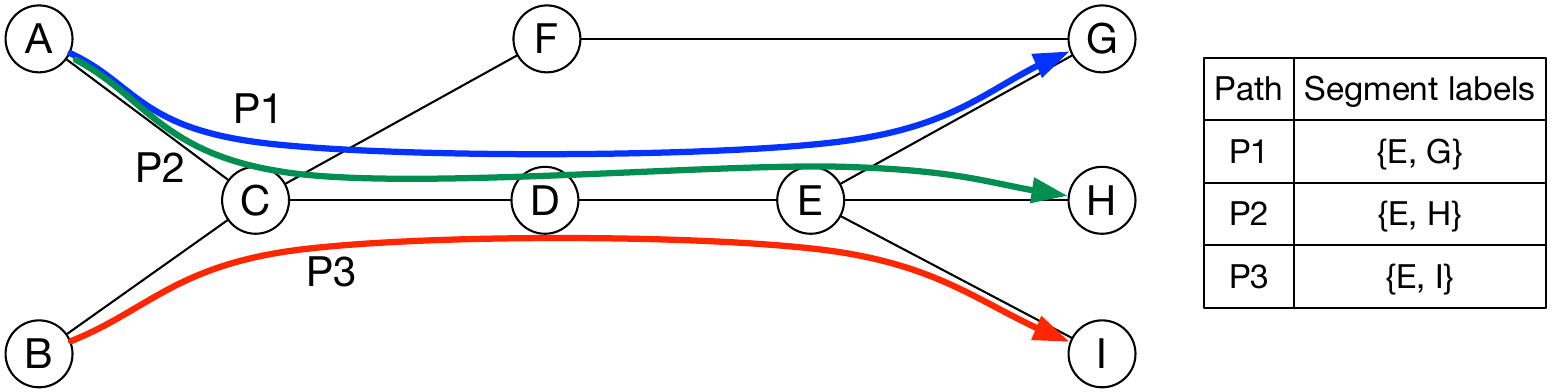}
% \vspace{-4mm}
\caption{An example where segment routing saves flow table entries. }
\label{fig:showcase}
% \vspace{-3mm}
\end{figure}

\begin{table}[h]
\centering
\label{table:rules}
\resizebox{0.7\columnwidth}{!}{%
\begin{tabular}{|c|c|c|}
\hline
Node & w.o. segment routing  & w. segment routing \\ \hline
A & 2 & 2 \\ \hline
B & 1 & 1 \\ \hline
C & 3 & 1 \\ \hline
D & 3 & 1 \\ \hline
E & 3 & 3 \\ \hline
\end{tabular}
}
\vspace{2mm}
\caption{Number of flow table entries for the example in Fig.~\ref{fig:showcase}. }
\end{table}

% \subsection{Prior work }
% \label{sec:related_SR}
% We briefly review prior work in segment routing.
Given the potential of segment routing, some recent work has started to investigate how to apply it in TE.
% Segment routing is a relative new technology and only has a limited number of research work.
In \cite{bhatia2015optimized},  the authors propose solutions for determining the optimal TE with segment routing and ECMP. In their scheme, they regard all nodes except for the source and the destination as candidate middlepoints, and split flows across exactly one middlepoint. Although they limit the number of middlepoints to one for each logical path, since they consider all the intermediate nodes for one demand, the search space for middlepoints is very large. The algorithm thus cannot scale to handle medium to large scale networks.
% For a large network with many demands, many nodes are selected as middlepoints to route them which still incurs significant overhead in terms of flow entries on the switches.
Hartert et al. \cite{hartert2015solving,HVSB15} studied a similar TE problem with segment routing under a constraint programming framework. Their middlepoint selection method also takes every node as a potential candidate on a per-demand basis, and they have to resort to heuristics to reduce the run time of the algorithm.

The exploration of segment routing in TE has been ad-hoc so far.
Existing work uses heuristics \cite{HVSB15} or some special for of segment routing \cite{bhatia2015optimized} without theoretical justification.
We are thus motivated to conduct a systematic study of segment routing in TE, including the theoretical characterization of the hardness results of various forms of segment routing, and practical algorithms of solving the problems.

Given our focus on TE, in the next section we review the two common types of TE formulations. We subsequently reveal an interesting connection between them, which we utilize later in the proof of several hardness results.

\section{Background on Traffic Engineering}
\label{sec:TE}
In our work, we focus on two common types of traffic engineering, depending on the objective criterion.  The first type maximizes the total throughput subject to the capacity and maximum demand constraints. Since it can be formulated as a maximum flow problem, we call it $TE_{MF}$. The second type minimizes the maximum link utilization, which acts as the system bottleneck. For this reason, we call it $TE_{LU}$. 

% In the remainder of this section, we review the two TE types, and reveal an interesting connection between them.

%{\color{red} probably need to say why this section is necessary. also shorten the discussion, reviewers should be familiar with most of the materials here already}

\subsection{Preliminaries}
\label{sec:TE-prelim}
Assume a directed graph $G=(V,E)$, where $V$ is the set of nodes and $E$ the set of directed edges. Given a node $v\in V$, $v^+$ denotes the set of outgoing edges of node $v$, i.e., the subset of edges in $E$ of the form $(v,u)$, $u\in V$. Similarly, the set $v^-$ denotes the set of incoming edges of $v$ of the form $(u,v)$, $u\in N$. The out-degree of $v$ is defined as the cardinality $|v^+|$, whereas the in-degree is defined as the cardinality $|v^-|$.

A \textit{flow network} $G=(V,E,c)$ is defined as a directed graph $G=(V,E)$, together with a non-negative function $c:V\times V\rightarrow\mathbb{R}_{\geq 0}$ that assigns to each edge $e\in E$ a non-negative capacity $c(e)$. If $(u,v)\not\in E$, then we define $c(u,v)=0$.

A \textit{walk} in a directed graph is an alternating sequence of vertices and edges, $v_0$, $e_0$, $v_1$, $\dots$, $v_{k-1}$, $e_{k-1}$, $v_k$, which begins and ends with vertices and has the property that each $e_i$ is an edge from $v_i$ to $v_{i+1}$. A \textit{path} is a walk where all edges are distinct. 
%A \textit{simple path} is a path where all vertices are distinct. 
The term $u$-$v$ path refers to any valid path from $u$ to $v$.

In flow networks, we usually distinguish between single-commodity and multi-commodity flows. For single-commodity flow, we consider a single commodity that consists of a source $s\in V$ and a sink $t\in V$, where $s\neq t$. %Our goal is to compute the maximum flow from $s$ to $t$ in the underlying flow network. 
For multi-commodity flows, we assume $L$ commodities of the form $(s_i,t_i)$, where $s_i,t_i\in V, s_i\neq t_i$. Each commodity $i$ is associated with a non-negative demand $D_i\geq0$. %%For convenience, we also use the notation $\bm{s}=(s_1,\dots,s_L)$ and $\bm{t}=(t_1,\dots,t_L)$, and write $(\bm{s},\bm{t})$ to denote the corresponding multi-commodity network. 
% the following is only for TE_MF
% Our goal is to maximize the sum of flow units that we send from sources $s_i$ to their corresponding sinks $t_i$.

\subsection{TE Type 1: $TE_{MF}$}
\label{sec:TE-MCF}
\begin{comment}
Since the single-commodity flow is a special case of the more general multi-commodity flow, we provide next several formulations for the latter. First, the maximum multi-commodity flow can be formulated as the following LP:

{\color{red} if the arc based formulation is not essential, can we just start with path-based formulation? it's quite well understood in TE literature}

\[
\begin{aligned}
\text{maximize}\qquad & \nu=\sum_{i=1}^L\sum_{e\in {s_i}^{+}}f_i(e)\\
\text{subject to }\qquad & \sum_{i=1}^L f_i(e)\leq c(e), \forall e\in E\\
& \sum_{e\in {s_i}^{+}}f_i(e)\leq D_i\\
& f_i(e)\geq 0, \forall e\in E, \forall i\in\{1,\dots,L\}\\
& \sum_{e\in u^+}f_i(e)=\sum_{e\in u^-}f_i(e), \forall i, \forall u, s_i\neq u\neq t_i\\
\end{aligned}
\]
\end{comment}

Let $\mathcal{P}_i$ be the set of all $s_i$-$t_i$ paths, and $\mathcal{P}_{i,e}$ the set of all $s_i$-$t_i$ paths that go through edge $e$. Then the maximum multi-commodity flow program can be expressed by the following path-based formulation:

\begin{align}
\text{maximize}\qquad & \nu=\sum_{i=1}^L\sum_{p\in \mathcal{P}_i}f_i(p)\label{TE-MCF-obj}\\
\text{subject to }\qquad & \sum_{i=1}^L\sum_{p\in \mathcal{P}_{i,e}} f_i(p)\leq c(e), \forall e\in E\label{TE-MCF-first-constraint}\\
& \sum_{p\in \mathcal{P}_i}f_i(p)\leq D_i\label{TE-MCF-second-constraint}\\
& f_i(p)\geq 0, \forall i\in\{1,\dots,L\}, \forall p\in \mathcal{P}_i\label{TE-MCF-last-constraint}
\end{align}
\begin{comment}
% In the above arc-based formulation, we divide the total flow into $k$ sub-flows, one for each commodity. The flow of commodity $i$ along edge $(u,v)$ is $f_i(u,v)$. The first constraint is a capacity constraint that the sum of all sub-flows on any edge cannot exceed the edge capacity. The second constraint imposes that the flow should be non-negative. The third constraint is a conservation constraint that ensures that each sub-flow is conserved in all nodes other than its source and its sink. For any valid flow $f$, the value of a flow $\nu(f)$ is defined as the total sum of units that all sub-flows $f_i$ send. The maximum flow is then simply defined as $\nu_{max}$. Note also that the multi-commodity flow formulation can also include upper (or lower) bounds for the demands of the commodities, which are easily expressed as linear constraints.
\end{comment}
% In the above formulation, we divide the total flow into $L$ sub-flows, one for each commodity. 
The commodity flow $i$ along path $p\in\mathcal{P}_i$ is $f_i(p)$. Constraint \eqref{TE-MCF-first-constraint} is a capacity constraint that the sum of all sub-flows on any edge cannot exceed the edge capacity. Constraint \eqref{TE-MCF-second-constraint} simply describes the maximum demand $D_i$ for commodity $i$. Finally, constraint \eqref{TE-MCF-last-constraint} simply imposes that the flow should be non-negative. For any valid flow $f$, the value of a flow $\nu(f)$ is defined as the total sum of units that all sub-flows $f_i$ send. The maximum flow is then simply defined as $\nu_{max}$.
$TE_{MF}$ is mostly used in data center backbone WANs \cite{JKMO13,HKMZ13}, where traffic is elastic and the main objective is to fully utilize the expansive WAN links.

Note that even though the single-commodity maximum flow accepts various combinatorial algorithms \cite{amo1993}, e.g., Ford-Fulkerson or Edmonds-Karp, there is to date no combinatorial algorithm for the maximum multi-commodity flow even though the problem is known to be strongly polynomial due to Tardos \cite{Tardos1986}. Furthermore, even though single-commodity networks always accept an integer maximum flow, this is not always the case with multi-commodity networks; in fact, the decision problem of integral multi-commodity flow is NP-complete even if the number of commodities is two, for both the directed and undirected cases \cite{eis1975}.

%Besides the path-based formulation, the maximum multi-commodity flow problem also accepts an arc-based formulation.
%The path formulation implicitly incorporates the node conservation constraints into the paths, so we do not need to list them explicitly as we did in the edge formulation. 
%The latter is more compact and can be represented in polynomial space, whereas the former can take space exponential to the input graph, since there can be exponentially many paths between two nodes.

\subsection{TE Type 2: $TE_{LU}$}
\label{sec:TE-LU}

$TE_{LU}$ is mostly used in carrier networks \cite{HVSB15,FT00}, where traffic demands are given and inelastic, and the main objective thus is to control the congestion or link utilization in order to ensure the smooth operation of the network. 
The general form for this type of TE is:
\begin{align}
\text{minimize}\qquad & \theta\label{TE-LU-obj}\\
\text{subject to }\qquad & \sum_{i=1}^L\sum_{p\in \mathcal{P}_{i,e}} f_i(p)\leq \theta\cdot c(e), \forall e\in E\label{TE-LU-first-constraint}\\
& \sum_{p\in \mathcal{P}_i}f_i(p)\geq D_i\label{TE-LU-second-constraint}\\
& f_i(p)\geq 0, \forall i\in\{1,\dots,L\}, \forall p\in \mathcal{P}_i\label{TE-LU-last-constraint}
\end{align}

The variable $\theta$ in the objective function \eqref{TE-LU-obj} refers to the maximum link utilization, which must be minimized. Constraint \eqref{TE-LU-first-constraint} ensures that $\theta$ will be at least as large as the maximum link utilization; constraint \eqref{TE-LU-second-constraint} ensures that each demand is satisfied; and the last constraint \eqref{TE-LU-last-constraint} is similar to $TE_{MF}$.% in Section~\ref{sec:TE-MCF}.

\section{Segment Routing with Shortest Paths}
\label{sec:opt} 

In this section, we consider TE with shortest path based segment routing. In detail, traffic is routed only along the shortest paths for a given segment. 
In this sense, our results provide  theoretical foundations for existing work that focuses on shortest path based segment routing \cite{bhatia2015optimized,HVSB15}. 
% In Section \ref{sec:model} we formulate the traffic engineering problem with segment routing using ECMP, whereas in Section \ref{sec:SRhardness} we discuss hardness results for acyclic segment routing.

%\subsection{Network model and TE}
%\label{sec:model}

%\draft{is there a way to simplify this and make all the notations coherent with the ones in \cref{sec:TE}? also do we need Table 2?}

%Our network is a directed graph $G=(V,E)$, where $E$ represents the set of directed links between switches.
%Each link $e\in E$ has a capacity $C_e$. 
%Let $F$ denote the set of flows, where each flow $f \in F$ is the aggregated traffic from an ingress to an egress switch.
%The demand of $f$ in a TE interval $D_f$ is given.

Assume there are in total $K$ middle points available. Each end-to-end path can use up to $M\leq K$ of these middle points. For a segment $s\in S$ between an ingress node and a middlepoint, two middlepoints, or a middlepoint and an egress node, there are multiple paths in general. We assume, for simplicity, that routing is done by ECMP over all shortest paths of a segment. This is consistent with prior work \cite{bhatia2015optimized}. We use  $T_i$ to denote the complete set of logical tunnels formed by segments in $S$ that can be used for commodity $i$, with up to $M$ middle points. A tunnel involves only ingress/egress switch, and the intermediate middlepoints. This can be constructed offline efficiently. 

Let $G_{t,s}$ denote if a tunnel $t$ uses segment $s$ or not, %$H_{s,p}$ denote if segment $s$ uses shortest path $p$ or not,
and $I_{p,e}$ denote if path $p$ uses link $e$ or not. Furthermore, let $\hat{P}_s$ be the set of all shortest paths for segment $s$,
% We use $P_f$ to denote the set of disjoint  paths corresponding to  flow $f$ where the paths in $P_f$ have the same source and destination with $f$. 
and $f_{i}(t)$ represent the flow in tunnel $t$ for commodity $i$. The split ratio  $x_{i,t}$ for $i$ on tunnel $t$ is defined as the ratio $x_{i,t}=\frac{f_i(t)}{\sum_{t \in T_i}f_{i}(t)}$.

The $TE_{LU}$ problem with segment routing can be formulated similar to Section~\ref{sec:TE-LU}, where the set of paths $\mathcal{P}_i$ for commodity $i$ is now replaced by the set of logical tunnels $T_i$: 
%\begin{align}%G_{t,s} H_{s,p}
%\min      & \quad\quad \theta \label{obj} %\\  
%\text{s.t. }& \sum_{f \in F} \sum_{t \in %T_f}\sum_{s\in S_t} \sum_{p\in P_s} x_{f,t}D_f\frac{I_{p,e}}{\lvert P_s \rvert}    \leq \theta\cdot C_e ,\forall e \in E,  \label{con:capacity} \\
%&  0 \leq x_{f,t}\leq 1 ,\forall  f  \in F,  t \in T_f, \label{con:nonnegative} \\
%& \sum_{t \in T_f}x_{f,t} = 1, \forall  f  \in F. \label{con:flow}
%\end{align}
\begin{align}%G_{t,s} H_{s,p}
\min      & \quad\quad \theta \label{obj} \\  
\text{s.t. }& \sum_{i=1}^{L} \sum_{t \in T_i}\sum_{s\in S_t}\sum_{p\in \hat{P}_s} f_{i}(t)\frac{I_{p,e}}{\lvert \hat{P}_s \rvert}    \leq \theta\cdot c(e) ,\forall e \in E,  \label{con:capacity} \\
&  0 \leq f_{i}(t) ,\forall  i  \in \{1,\dots,L\},  t \in T_i, \label{con:nonnegative} \\
& \sum_{t \in T_i}f_{i}(t) \geq D_i, \forall  i  \in \{1,\dots,L\}. \label{con:flow}
\end{align}
%The objective is again to minimize the maximum link utilization $\theta$ across the network.
The capacity constraint \eqref{con:capacity} indicates that the total traffic routed to link $e$ from across all flows, tunnels, segments, and shortest paths, cannot exceed $\theta$ times the link capacity. Since ECMP is used for routing within any segment $s$, each shortest path $p$ of segment $s$ receives flow equal to $f_{i}(t)/\lvert \hat{P}_s \rvert$. 
%Constraint \eqref{con:flow} ensures that the entire demand for $f$ is routed.

Regarding the TE asymptotic complexity, we have the following result when $M$ is fixed and not part of the input:

\begin{proposition}
\label{prop:te-sp-lu}
For fixed $M$ with respect to the input graph $G$, the $TE_{LU}$ problem described by Equations \eqref{obj}-\eqref{con:flow} can be solved in (weakly) polynomial time.
\end{proposition}
\begin{proof}
The number of commodities $L$ cannot exceed $|V|\cdot(|V|-1)$, and the number $|T_i|$ of tunnels per commodity $i$ is upper bounded by ${{K}\choose{0}}+\cdots+{{K}\choose{M}}$, where $K\leq|V|$. For fixed $M$ w.r.t. the input graph $G$, $|T_i|$ has polynomial size w.r.t the graph. Finally, the number $S_t$ of segments per tunnel cannot exceed $K+1\leq|V|+1$, since a tunnel can use at most all $K$ middlepoints. For the inner sum $\sum_{p\in \hat{P}_s} \frac{I_{p,e}}{\lvert \hat{P}_s \rvert}$, note that it basically denotes the percentage of shortest paths for segment $s$ that use link $e$. However, this can be computed in polynomial time, e.g. by using the techniques in \cite{B2001}.

Thus, we have proved that for fixed $M$, the LP has a polynomial number of variables, and a polynomial number of constraints whose coefficients can be computed in polynomial time. The proposition then immediately follows by standard results in linear programming \cite{K1984,K1980}.
\end{proof}

Given that the $TE_{MF}$ formulation is very similar, we can similarly prove that is can also be solved in (weakly) polynomial time for segment routing with shortest paths.

Finally, we observe that the TE problem is naturally related to the shortest path centrality that we discuss in Section \ref{sec:2pass}. Indeed, the inner part $\sum_{p\in \hat{P}_s} \frac{I_{p,e}}{\lvert \hat{P}_s \rvert}$ of constraint \eqref{con:capacity} precisely describes the percentage of shortest paths for segment $s$ that use a specific edge, and note that we do that for all possible segments. Even though shortest path centrality refers to a node rather than an edge and equally takes into account all possible source-destination pairs, constraint \eqref{con:capacity} reveals interesting connections between the popular centrality metric and the segment routing problem.
Furthermore, this hints at the potential power of shortest path centrality and it variants to the middlepoint selection problem, which we confirm experimentally in Section \ref{sec:evaluation}.

\section{Centrality Based Middlepoint Selection}
\label{sec:2pass}

%The previous sections investigated the fundamentals of segment routing, and showed that the TE for unrestricted segment routing is NP-hard. On the other hand, 
Proposition \ref{prop:te-sp-lu} suggests that if we only allow a fixed number $M$ of middlepoints per path, then TE with shortest paths is (weakly) polynomially computable.
One approach would then be to consider all nodes as candidate middlepoints, i.e. $K=|V|$. However, this results in very large TE programs that are costly to solve. An alternative is to just consider a small number of middlepoints such that $K\ll|V|$, that would still produce good output for the TE. Given that this is generally NP-hard \cite{HVSB15}, in this section we discuss practical middlepoint selection based on alternative centrality measures with polynomial complexity. Note that these centralities are \textit{structural} metrics that look at the graph structure, i.e., the connections among the various nodes. However, they generally do not take into account the flow network and its flow conservation and capacity constraints.%, which was the case with the NP-hard flow centralities in Appendix~\ref{sec:appendix}.

{\bf Shortest-path centrality.}
We start with \textit{shortest-path} centrality, which characterizes the power of a node in terms of the number of shortest paths that go through that node for a randomly picked source-destination pair. Concretely, assume a directed graph $G=(V,E)$. The shortest-path betweenness centrality of a node $v\in V$ \cite{F1977} is defined as:
\begin{equation}
\delta(v) = \sum_{s,t\in V|s\neq v\neq t}\frac{\sigma_{st}(v)}{\sigma_{st}},
\end{equation}
where $\sigma_{st}(v)$ is the number of shortest paths from $s$ to $t$ that go through $v$, and $\sigma_{st}$ the total number of shortest paths from $s$ to $t$. Calculating the shortest path centrality of all vertices in a graph requires $\Theta({|V|}^3)$ time and $\Theta({|V|}^2)$ space. This can be achieved by augmenting the Floyd-Warshall algorithm for the all-pairs shortest-paths problem with path counting. Brande's algorithm improves these bounds by only using $O(|V|+|E|)$ space and running in $O(|V|+|E|)$ and $O(|V|\cdot|E|+{|V|}^2\log|V|)$ time on unweighted and weighted networks, respectively \cite{B2001}.

{\bf Group shortest-path centrality.}
As opposed to the aforementioned \textit{individual} centrality, the \textit{group} shortest-path betweenness centrality of a group of nodes $C\subseteq V$ refers to the combined centrality of the group \cite{eb1999}. It is defined as:
\begin{equation}
\delta_{\mathcal{G}}(C)(v) = \sum_{s,t\in V|s\neq v\neq t}\frac{\sigma_{st}(C)}{\sigma_{st}},
\end{equation}
where $\sigma_{st}(C)$ the number of shortest paths that go through \textit{any node} in $C$.
Group betweenness centrality can be approximated within a factor $1-\frac{1}{e}$ to the optimal using a greedy incremental algorithm \cite{depz2009}. Brandes' algorithm for computing the betweenness centrality of all vertices can be modified to compute the group betweenness centrality of one group of nodes with the same asymptotic running time \cite{pyeb2009}.

{\bf Degree centrality.}
A simple alternative to the family of shortest path centralities is \textit{degree} centrality. The degree centrality of a node $v\in V$ is defined as the average of its in-degree and its out-degree:
\begin{equation}
d(v) = \frac{|v^+|+|v^-|}{2}
\end{equation}
Degree centrality captures a node's power by its number of neighbors; the higher that number, the better connected the node and the larger its centrality. Despite its simplicity, degree centrality can capture to a good extent a node's structural importance.

{\bf Weighted centralities.}
All aforementioned centralities only employ the graph connectivity information, and treat all links equally. However, in practice links are further characterized by their capacity. We can thus define variants of the previous centralities that additionally take into account the link capacity information. A simple approach is to associate each edge with the non-negative cost $\frac{1}{c(e)}$. This is based on the observation that the higher the capacity, the lower the cost of the link since it can accommodate larger flows. The shortest path centrality variants are simple to define, if we note that the cost of a path is the sum of the costs of its constituent links, and the shortest path refers to the path with the minimum cost among all paths. In a similar spirit, we can define the weighted degree of any node as the sum of the costs of the edges that are incident to the node. Intuitively, we expect that the weighted variants should perform better since they take into account both the connectivity and the capacity information. Section~\ref{sec:eva_weighted} empirically confirms our intuition.

%Finally, note that all aforementioned centralities can accept a weighted variant, where each edge is characterized by a non-negative weight. A common approach is to associate each edge with a unit cost, so that a shortest path minimizes the number of hops from the source to the destination. One shortcoming with this definition is that if edge capacities differ dramatically, then the number of hops may not represent well how good a path is. For this reason, an alternative cost function is $\frac{1}{c_e}$, where $c_e$ the edge capacity. The idea is that the higher the capacity, the lower the cost of the link since it can accommodate a larger flow. In the special case where all link capacities are identical, the two definitions are equivalent up to the multiplicative constant $\frac{1}{c_e}$.
% \input{algorithm}
%!TEX root = main.tex

\section{Evaluation}
\label{sec:evaluation}

In this section, we conduct trace-driven simulations to evaluate the performance of centrality based middlepoint selection methods.
The experiments are designed to answer the following important questions:
\begin{itemize}[noitemsep]
\item What is the best parameter setting for centrality based middlepoint selection? Specially, how many middlepoints per commodity, and how many middlepoints in total should we use?
\item How does our centrality based approach compare to existing work, in terms of both performance and complexity?
\item How do various centrality definitions perform against each other?
\end{itemize}

\subsection{Methodology}
\label{sec:methodology}

We use two network topologies from the dataset provided by DEFO (Declarative and Expressive Forwarding Optimizer) \cite{defo} which is used in \cite{HVSB15}. %The average node degree of topologies are about 4 or 5.
One is a synthetic network with 100 nodes ({\tt synth100} in the dataset), and the other is a real network with 161 nodes ({\tt rf3257} in the dataset). Table~\ref{table:dataset} provides more details about the networks.
The DEFO dataset also contains information of commodity flows (simply referred to as flows hereafter) for these topologies. For the real 161-node topology the flows are provided by the ISP \cite{HVSB15}. For the synthetic topology the demand matrices are computed using the approach in \cite{M05}. As explained in \cite{HVSB15}, this approach uses a gravity model fed with i.i.d. exponential random variables. It produces realistic demand matrices as shown in \cite{M05,HVSB15}.
\begin{table}[!h]
\centering
\resizebox{0.8\columnwidth}{!}{%
\begin{tabular}{|l|l|l|l|l|}
\hline
Type&ID&\# nodes&\# links&\# flows\\ \hline
% Synthetic&synth50&50&276&2449 \\ \hline
Synthetic&{\tt synth100} &100&572&9817 \\ \hline
Real&{\tt rf3257} &161&656&25486 \\ \hline
\end{tabular}
}
\vspace{3mm}
\caption{Dataset Summary.}
\label{table:dataset}
% \vspace{-3mm}
\end{table}

We perform simulations on servers each with a 2.2~GHz 64-bit 8-Core Xeon processor and 128~GB memory. We use the {\tt cvxpy} \cite{cvxpy} modeling language with the {ECOS solver} \cite{D13} to solve the LPs.
Our evaluation compares the following schemes:
\begin{itemize}[noitemsep]
\item {\em Baseline}: Traditional approach of applying segment routing studied in Sec.~IV of \cite{bhatia2015optimized}. Specifically, the TE problem assumes that every node is a candidate of middlepoint, and exactly one middlepoint is used for each flow. Only shortest paths are used for segment routing.
\item {\em Random}: Our approach where a total of $K$  middlepoints are randomly selected and used in TE problems \temf~and \telu.
\item {\em Shortest-path centrality (SP)}: Our approach where middlepoints are selected using shortest-path centrality as explained in \cref{sec:2pass} for TE.
\item {\em Group shortest-path centrality (GSP)}: Our approach where middlepoints are selected using group shortest-path centrality.
\item {\em Degree centrality (Degree)}: Our approach where middlepoints are selected using degree centrality.
\end{itemize}

\subsection{Microscopic Performance }
\label{sec:microscopic}

First we aim to understand the microscopic performance of our centrality based approach.
There are two key parameters affecting our approach in general: the number of middlepoints per flow $M$, and the total number of available middlepoints for all flows $K$. Their effects need to be thoroughly understood before we compare our approach to existing methods.

\noindent{\bf Number of middlepoints per flow $M$}:
We begin by trying to answer how many middlepoints should be used for each flow. Note that there is an inherent tradeoff: more middlepoints per flow leads to more flexibility in constructing the paths and balancing the traffic, and thus better performance. On the other hand it also means more overhead in terms of higher complexity of the TE algorithms.

To demonstrate this tradeoff, we use \telu~and compute the maximum link utilization with our centrality based middlepoint selection when the number of middlepoints per flow $M$ is equal to 1 or 2. We use the {\tt synth100} network with 100 nodes. We choose 1000 flows for the topology randomly for ten times and report the average. We apply {\em GSP} to select the middlepoints, and vary the total number of available middlepoints $K$ from 2 to 6. For a given $K$, the flows are identical for different values of $M$.
Note we also experiment with \temf~and the 161-node network; the results are qualitatively similar and omitted here for space.

\begin{figure}[thb]
\centering
\begin{minipage}{0.48\linewidth}
\includegraphics[width=\linewidth]{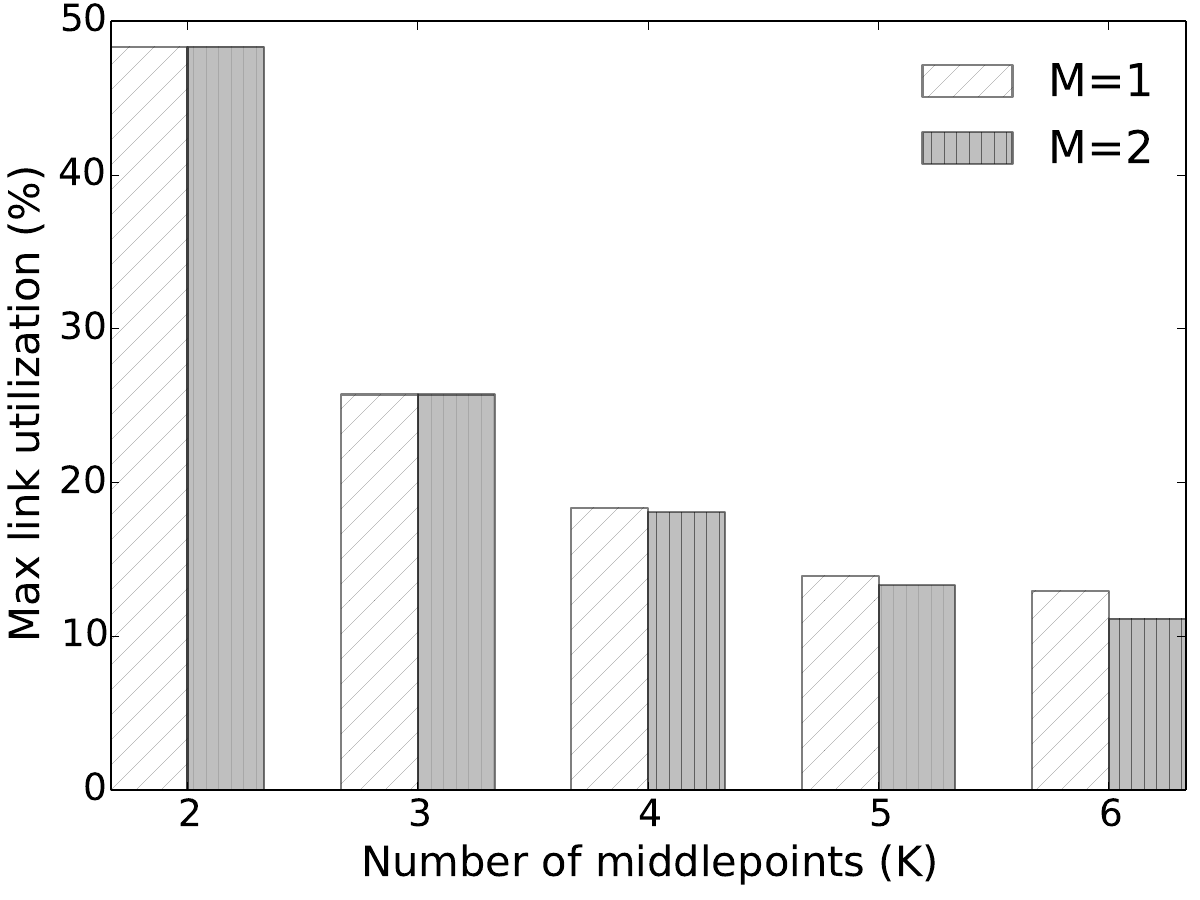}
      \caption{Maximum link utilization of the 100-node network with 1000 flows and varying $M$.}
      \label{fig:ul_path_k100}
      \end{minipage}
\begin{minipage}{0.48\linewidth}
   \includegraphics[width=1\linewidth]{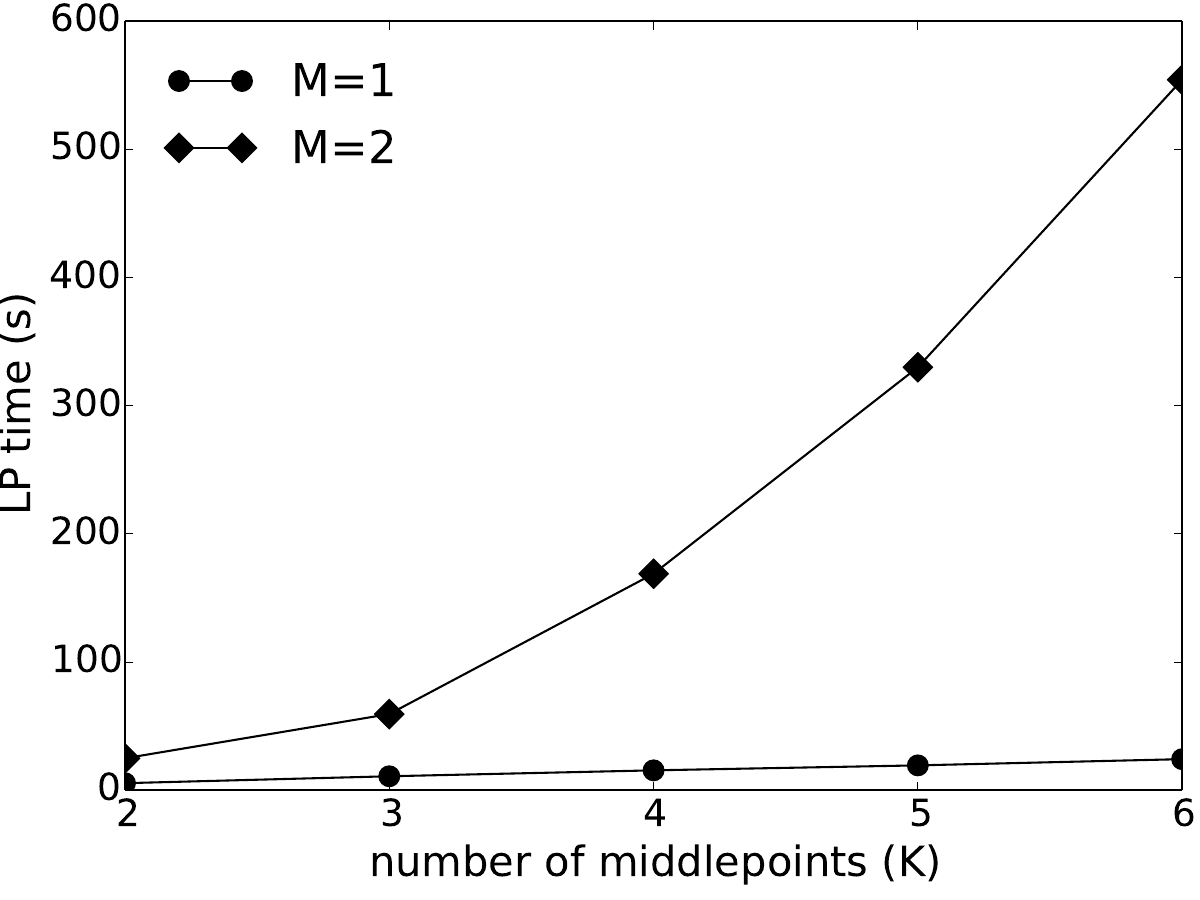}
      \caption{LP solving time of \telu~on the 100-node network with 1000 flows and varying $M$.}
      \label{fig:time_path_k100}
      \end{minipage}
% \vspace{-3mm}
\end{figure}

Fig.~\ref{fig:ul_path_k100} and Fig.~\ref{fig:time_path_k100} depict the results. We find that interestingly, the maximum link utilizations for $M=1$ and $M=2$ are quite similar. Yet the time of solving the TE with $M=2$ is much higher than $M=1$ as shown in Fig.~\ref{fig:time_path_k100} (a difference of almost 20x) when $K=6$). Given that the middlepoints are central, there is indeed a low probability that a bottleneck link exists between two middlepoints. Hence, maximum link utilization is largely determined by the bottleneck links between either the source and the middlepoint, or between the middlepoint and the destination. Therefore we conclude that 1 middlepoint per flow is good enough for performance, and use that throughout the remainder of the experiments.

\noindent{\bf Number of total middlepoints $K$}:
We next run experiments to verify that just a few central middlepoints are sufficient to achieve satisfactory TE performance.
% If more middlepoints are selected as candidates for TE, the link utilization will also decrease. Yet more middlepoints also means the centrality algorithms need to run longer. Thus we design some experiments to verify this.
We vary the total number of middlepoints $K$ from 1 to 6 for \telu~and from 1 to 8 for \temf. The middlepoints are selected using {\em Random}, {\em SP}, {\em GSP}, and {\em Degree} as explained in \cref{sec:methodology}.
We use both the 100-node and 161-node networks, and randomly choose 1000 flows and 2000 flows respectively for 10 runs. We report the average and standard deviation results.
Since {\em Random} is non-deterministic, we randomly select 5 sets of middlepoints for each of the 10 flow sets, resulting in 50 runs in total for {\rm Random}.
For \temf, in order to make the results more readable, we scale the traffic volumes by 10 times for 100-node topology and 40 times for 161-node topology, respectively.

We depict the results in Fig.~\ref{fig:ul_middle_points100}--Fig.~\ref{fig:161_multi40_2000_middle}.
As expected, more available middlepoints improve TE performance.
For \telu, when there is only one available middlepoint for the network, every flow has to be routed through it, which severely limits the path choice and the maximum link utilization is way above 1 for {\em Random} and over 1 for the other schemes. With two middlepoints the maximum link utilization is dramatically reduced by over 50\% for most schemes as seen in Fig.~\ref{fig:ul_middle_points100} and Fig~\ref{fig:ul_middle_points161}.
The same can be observed for \temf~in Fig.~\ref{fig:100_multi10_1000_middle} and Fig.~\ref{fig:161_multi40_2000_middle}.
The demand satisfaction ratio is improved by around a factor of 2 when $K$ increases to 2.

\begin{figure}[h]
\centering
\begin{minipage}{0.49\linewidth}
   \includegraphics[width=\linewidth]{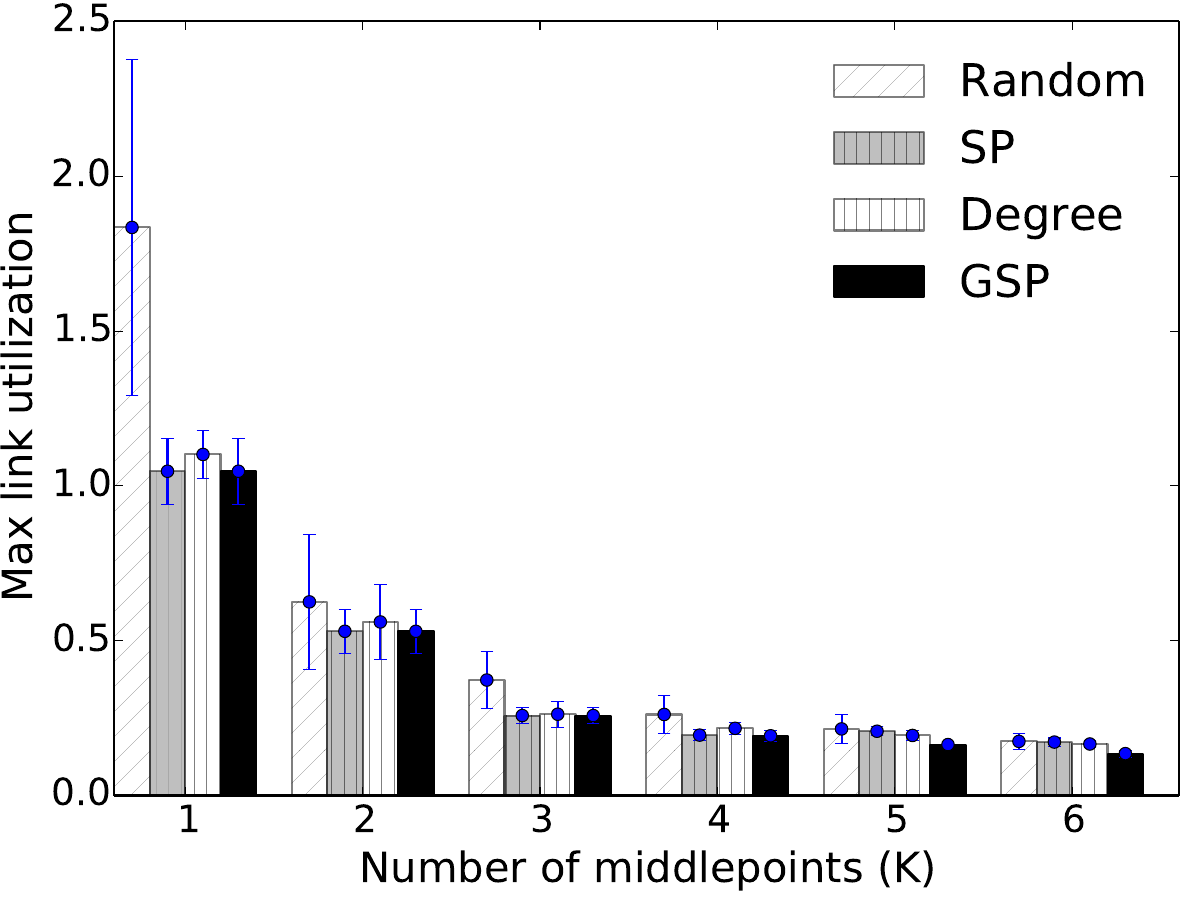}
      \caption{100-node network with 1000 flows of \telu.}
      \label{fig:ul_middle_points100}
\end{minipage}
\begin{minipage}{0.49\linewidth}
\includegraphics[width=\linewidth]{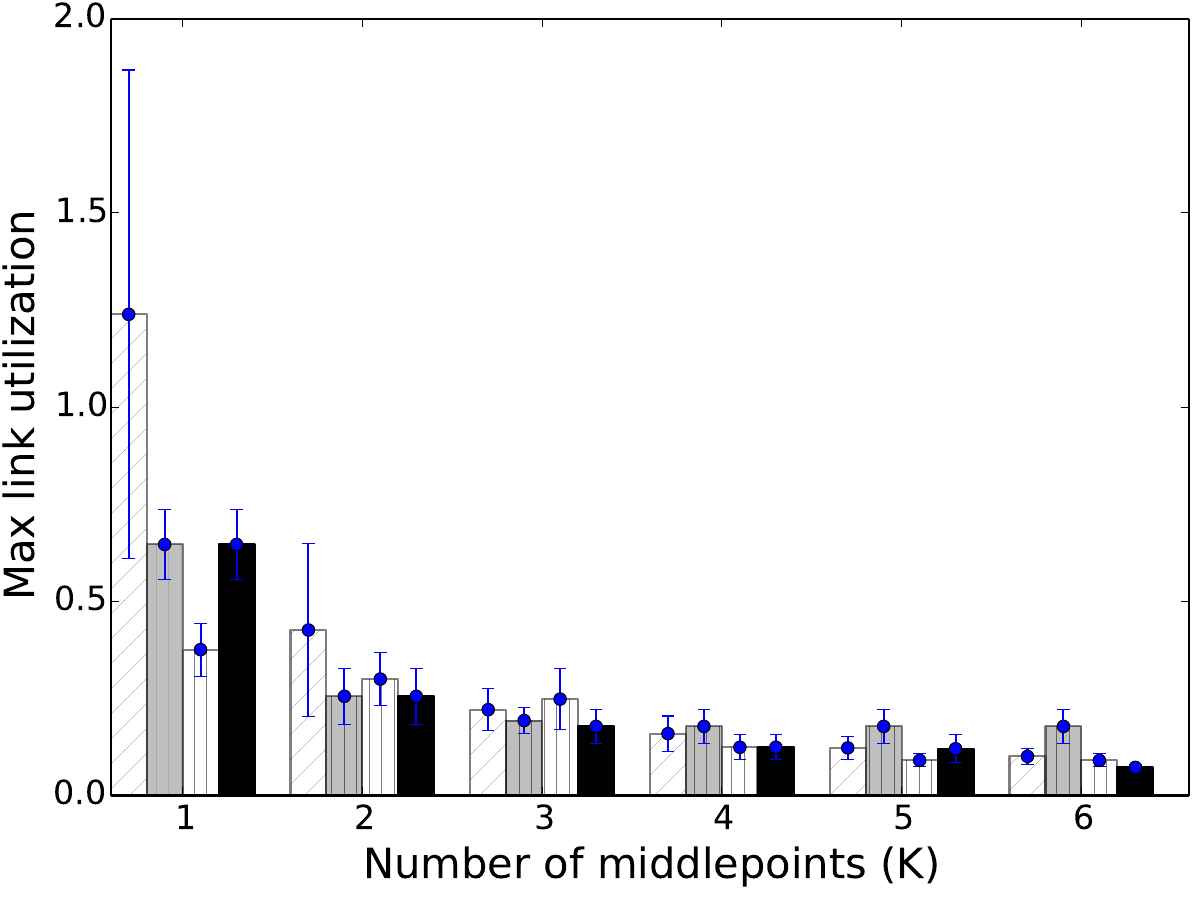}
   % where an filename suffix will be assumed under latex, and a
   % .pdf suffix will be assumed for pdflatex %
   %\subsection{Sequential Hypothesis Testing}
      \caption{161-node network with 2000 flows of \telu.}
      \label{fig:ul_middle_points161}
\end{minipage}
% \vspace{-3mm}
\end{figure}

\begin{figure}[h]
\centering
\begin{minipage}{0.49\linewidth}
   \includegraphics[width=\linewidth]{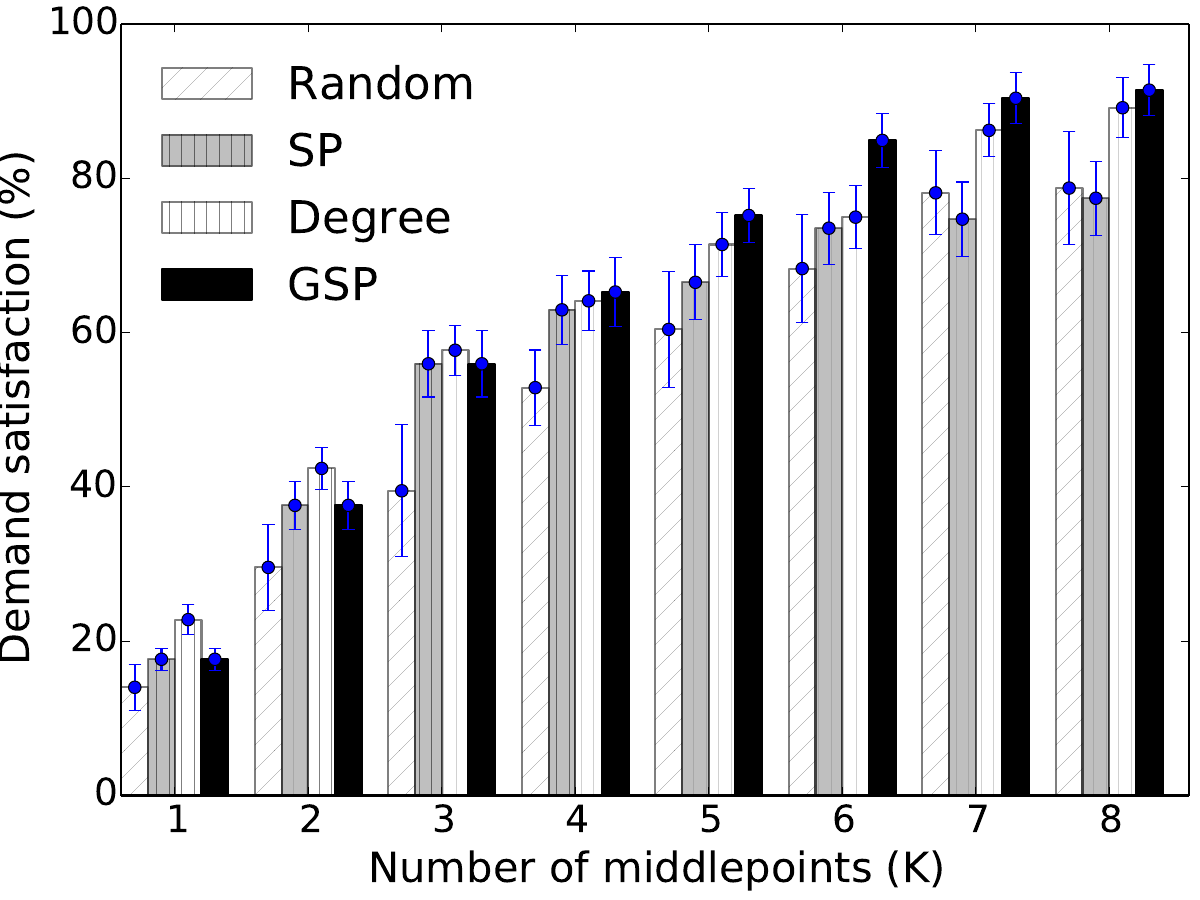}
   % where an filename suffix will be assumed under latex, and a
   % .pdf suffix will be assumed for pdflatex %
   %\subsection{Sequential Hypothesis Testing}
      \caption{100-node network with 1000 flows of \temf.} \label{fig:100_multi10_1000_middle}
\end{minipage}
\begin{minipage}{0.49\linewidth}
\includegraphics[width=\linewidth]{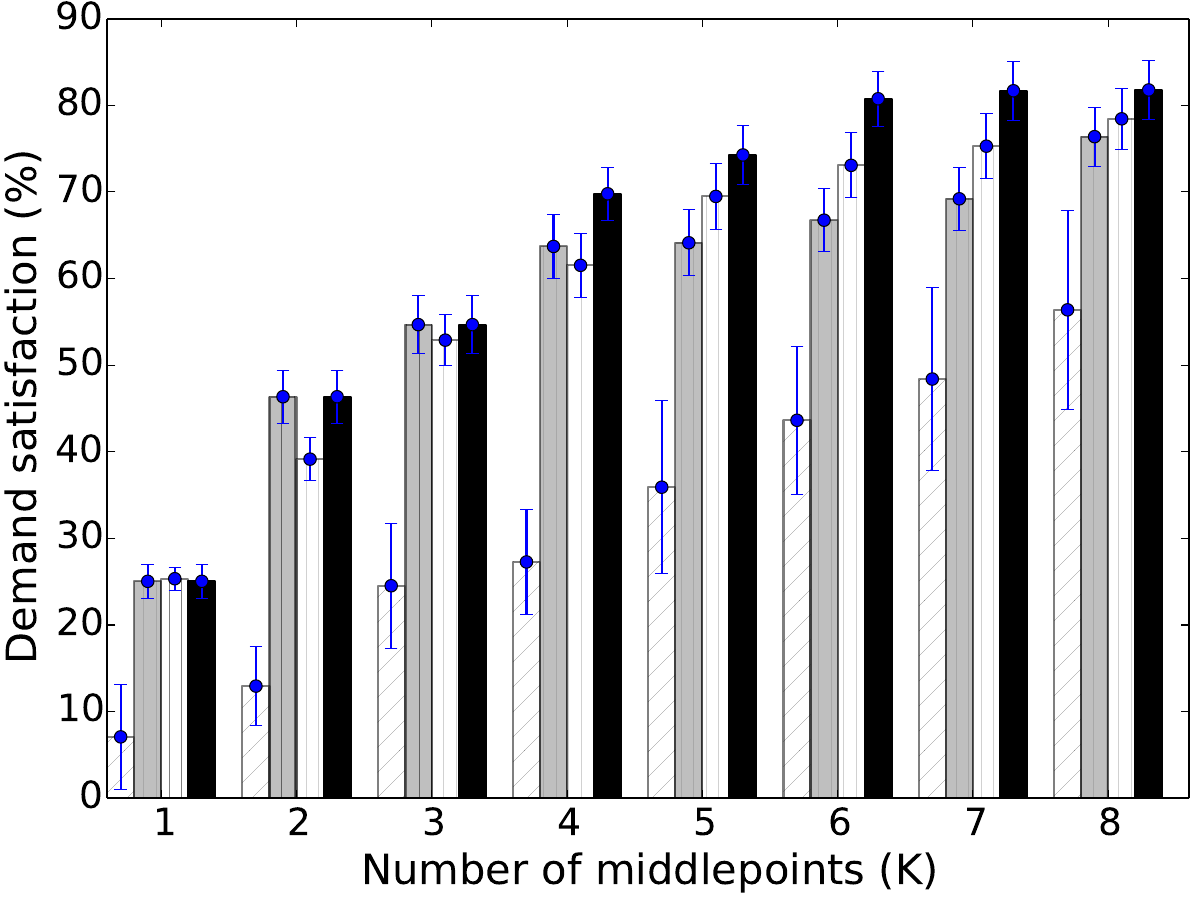}
   % where an filename suffix will be assumed under latex, and a
   % .pdf suffix will be assumed for pdflatex %
   %\subsection{Sequential Hypothesis Testing}
      \caption{161-node network with 2000 flows of \temf.} \label{fig:161_multi40_2000_middle}
\end{minipage}
\vspace{-3mm}
\end{figure}

Another important observation is that the TE performance exhibits diminishing marginal gains as $K$ increases.
For \telu, when more than 4 middlepoints are used, very limited gains are observed ($<$10\%) in Fig.~\ref{fig:ul_middle_points100} and Fig.~\ref{fig:ul_middle_points161}.
For \temf, beyond 7 middlepoints there is little demand satisfaction improvement especially for {\em GSP} in Fig.~\ref{fig:100_multi10_1000_middle} and Fig.~\ref{fig:161_multi40_2000_middle}.
On the other hand the runtime of the TE algorithms increases dramatically due to the growing size of the LP problems. Take the 161-node network for instance. The LP time for \telu~increases by $\sim$50\% when $K$ increases from 4 to 6 as shown in Table~\ref{table:time_lu}, and from 6 to 8 for \temf~as in Table~\ref{table:time_mf}.

\begin{table}[!h]
\centering
\resizebox{0.8\columnwidth}{!}{%
\begin{tabular}{|l|c|c|c|c|c|c|}
\hline
\backslashbox{Scheme}{$K$}&1&2&3&4&5&6 \\ \hline
% Synthetic&synth50&50&276&2449 \\ \hline
Random&11.25&24.21&36.50&51.29&66.00&81.48 \\ \hline
SP&9.79&19.56&29.35&38.31&46.27&54.25 \\ \hline
Degree&8.92&17.83&28.49&42.98&55.45&64.96 \\ \hline
GSP&9.75&19.54&29.30&39.64&54.85&70.81 \\ \hline
\end{tabular}
}
\vspace{3mm}
\caption{Average LP time (seconds) of 161-node network with 2000 flows of \telu.}
\vspace{-3mm}
\label{table:time_lu}
\end{table}

\begin{table}[!h]
\centering
\resizebox{0.98\columnwidth}{!}{%
\begin{tabular}{|l|c|c|c|c|c|c|c|c|}
\hline
\backslashbox{Scheme}{$K$}&1&2&3&4&5&6&7&8 \\ \hline
% Synthetic&synth50&50&276&2449 \\ \hline
Random&12.10&23.87&32.78&45.91&60.04&76.43&100.02&110.00 \\ \hline
SP&10.63&19.49&26.46&35.66&43.36&53.80&72.30&82.26 \\ \hline
Degree&9.92&17.68&26.25&41.21&52.52&60.47&78.97&86.62 \\ \hline
GSP&10.42&20.00&26.83&38.35&52.79&67.49&81.44&91.88 \\ \hline
\end{tabular}
}
\vspace{3mm}
\caption{Average LP time (seconds) of 161-node network with 2000 flows of \temf.}
\label{table:time_mf}
\vspace{-3mm}
\end{table}

Based on the above results, we conclude that 4 middlepoints for \telu   ~and 7 middlepoints for \temf~are the sweetspots of the tradeoff between performance and complexity. We thus use these settings in the rest of the experiments.
% Thus all middlepoint selection algorithms output 4 nodes. Note that this holds for all three networks we considered with different scales.
This confirms the intuition behind our centrality based approach, namely, that it suffices to just use a small fraction of nodes as middlepoints (2.48\%--7\% of nodes) to achieve satisfactory performance.

\subsection{Comparison with Baseline}
\label{sec:vsTE}

Our motivation of using centrality based middlepoint selection is to reduce the high complexity of existing approaches, which takes  all nodes in the network as middlepoints \cite{bhatia2015optimized} as discussed in \cref{sec:intro}.
% This means that for 100-node topology for example, there are 98 candidate middlepoints for each demand.
We now compare our approach against {\em Baseline} to validate its effectiveness in this regard. The experiments here are performed on the 100-node topology for both TE formulations.
The maximum link utilization and LP time of \telu~are shown in Fig.~\ref{fig:baseline_LU_util} and Fig.~\ref{fig:baseline_LU_LP}, respectively. The demand satisfaction ratio and corresponding LP time are depicted in Fig.~\ref{fig:baseline_MF_per} and Fig.~\ref{fig:baseline_MF_LP}, respectively, for \temf.
We scale the demands of flows by a factor of 2 for \telu~and a factor of 40 for \temf.

% But it needs nearly three hours for LP time. So we just run from 500 to 700 demands to make sure the LP time is less than 1000 seconds.

\begin{figure}[h]
\centering
\begin{minipage}{0.49\linewidth}
   \includegraphics[width=\linewidth]{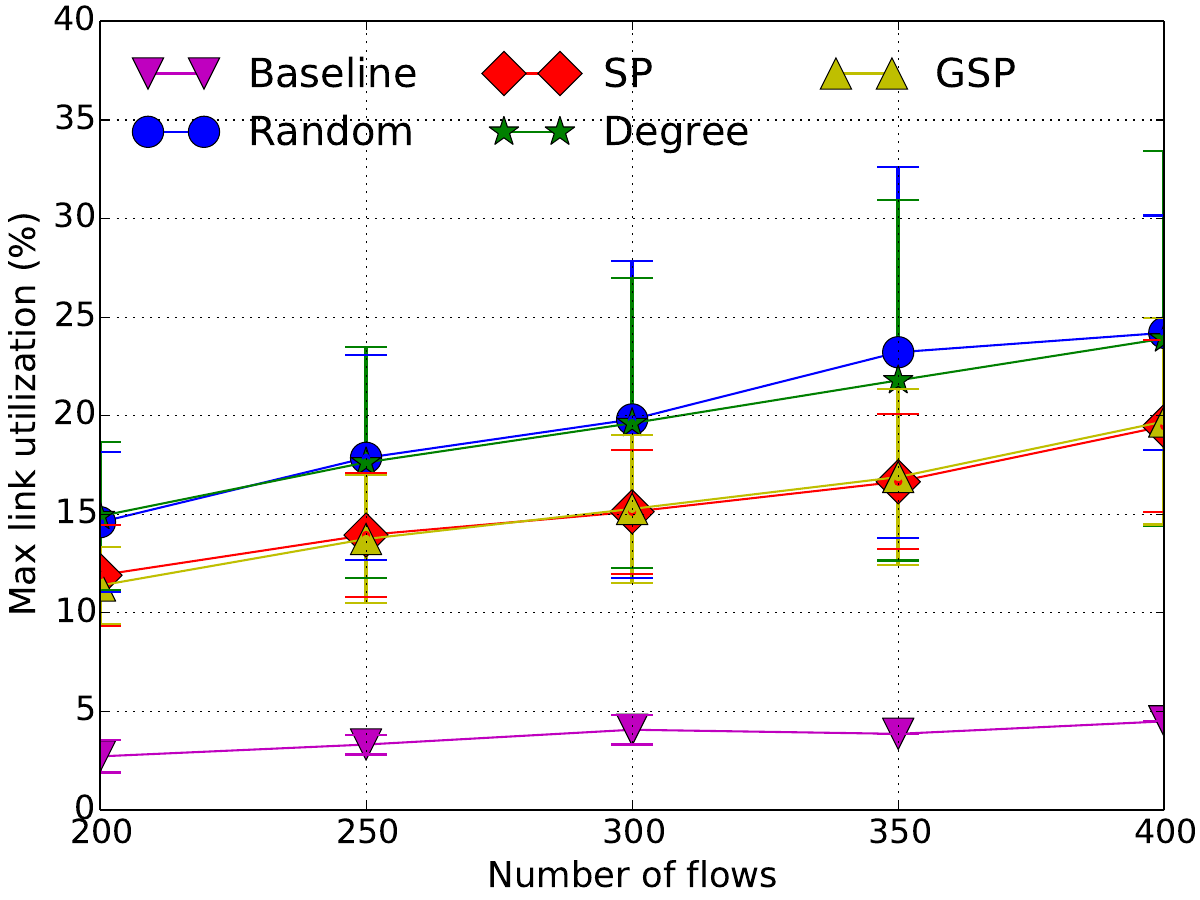}
      \caption{Performance of \telu~with different centralities and {\em Baseline}.} \label{fig:baseline_LU_util}
\end{minipage}
\begin{minipage}{0.49\linewidth}
\includegraphics[width=\linewidth]{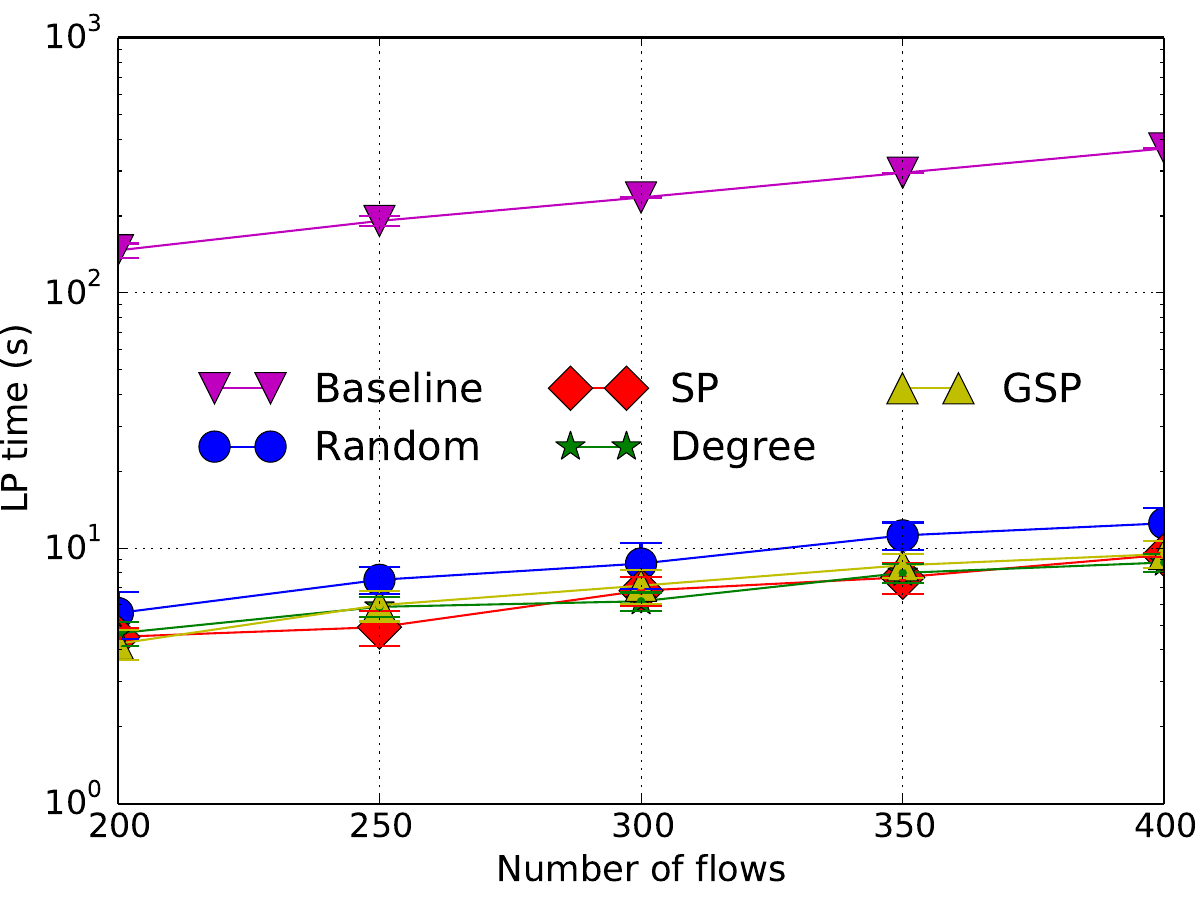}
      \caption{LP time of \telu~with different centralities and {\em Baseline}. Note the log scale of the y-axis.} \label{fig:baseline_LU_LP}
\end{minipage}
\vspace{-3mm}
\end{figure}

\begin{figure}[h]
\centering
\begin{minipage}{0.49\linewidth}
   \includegraphics[width=\linewidth]{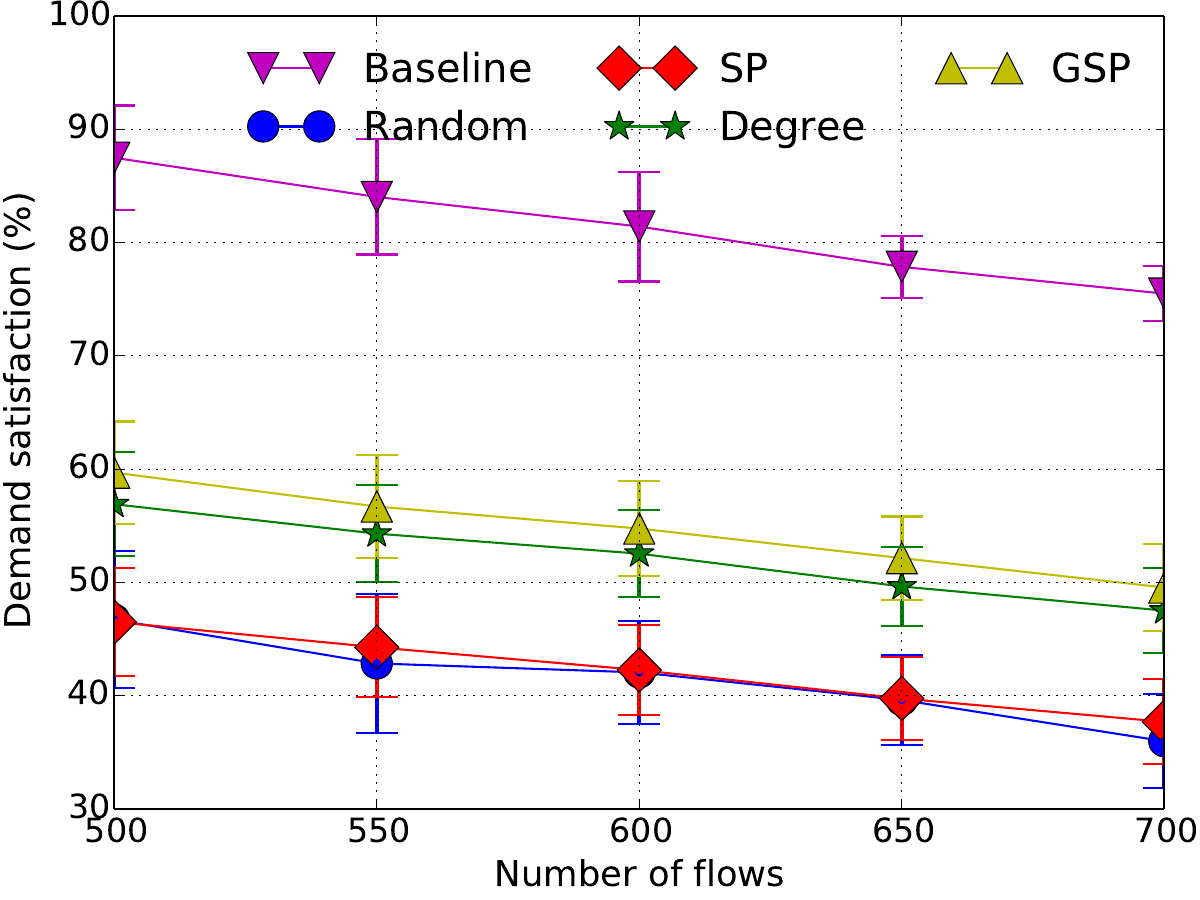}
      \caption{Performance of \temf~between different centralities and {\em Baseline}.} \label{fig:baseline_MF_per}
\end{minipage}
\begin{minipage}{0.49\linewidth}
\includegraphics[width=\linewidth]{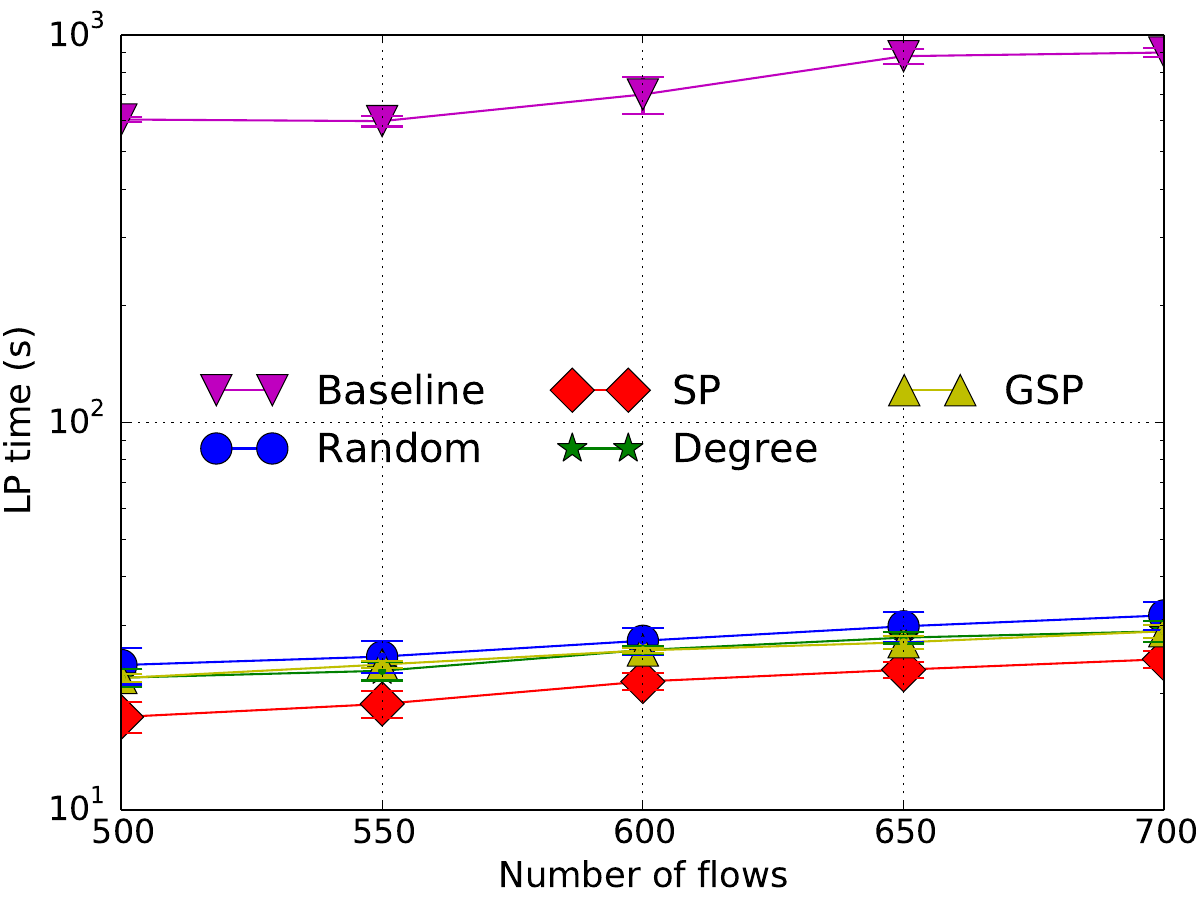}
      \caption{LP time of \temf~between different centralities and {\em Baseline}. Note the log scale of the y-axis.} \label{fig:baseline_MF_LP}
\end{minipage}
% \vspace{-3mm}
\end{figure}

Notice that with {\em Baseline}, the TE problems have much more variables and constraints due to the large number of middlepoints. As a result, our machines can only solve \telu~with $\sim$400 flows, and \temf~with $\sim$1500 flows.
{For problems beyond these scales the solver reports error messages.}
Recall that with centrality based middlepoint selection the solver can easily handle problems with 3000 flows even for the bigger 161-node topology as we will show in \cref{sec:whichcentrality}. In addition, solving \temf~with {\em Baseline} and 1500 flows takes more than three hours, far exceeding the time scale (5--10 min) at which TE is performed in practice \cite{JKMO13,HKMZ13,HVSB15}. Thus we only run {\em Baseline} with up to 700 flows for \temf~to make sure the LP time is less than 1000 seconds.

As shown in Fig.~\ref{fig:baseline_LU_util}, the maximum link utilization of our approach is about 4--5 times that of {\em Baseline}, whereas the LP time of {\em Baseline} is at least 40 times worse than any centrality based approach shown in Fig.~\ref{fig:baseline_LU_LP}.
For \temf, the demand satisfaction ratio of {\em Baseline} is about 1.5 times of ours in Fig.~\ref{fig:baseline_MF_per} but the LP time is about 60 times higher than ours as in Fig.~\ref{fig:baseline_MF_LP}.

Indeed we observe that our centrality based approach sacrifices performance in order to reduce the complexity of TE.
We argue that this is a sensible tradeoff to make in most cases, especially for data center backbone WANs that use \temf~with very short time periods of 5--10 min \cite{JKMO13,HKMZ13,HVSB15}. Centrality based approach can support much larger topologies and much more flows with orders of magnitude smaller runtime. One can also increase $K$ to obtain better performance if necessary.

\subsection{Comparison of Various Centralities}
\label{sec:whichcentrality}

We now wish to understand the relative performance of various centralities in realistic settings.
We use both the 100-node and 161-node topologies with $M=1$. Total number of middlepoints $K$ is set to 4 for \telu~and 7 for \temf~based on our previous experiments.
We vary the number of flows and for a given number of flows randomly draw flows 15 times from the traces.
For {\em Random} we perform 5 independent random selections of middlepoints for a given set of flows, resulting in 75 runs in total. For each run we compute the respective performance metrics and report the average and standard deviation.
In order to make the results more readable we scale the demands by 10 for 100-node topology and 40 for 161-node topology, respectively for \temf.

% In this section we first investigate the four centralities without using weights, and later focus on their weighted versions for completeness.

\begin{figure}[htpb]
\centering
\begin{minipage}{0.49\linewidth}
   \includegraphics[width=1\linewidth]{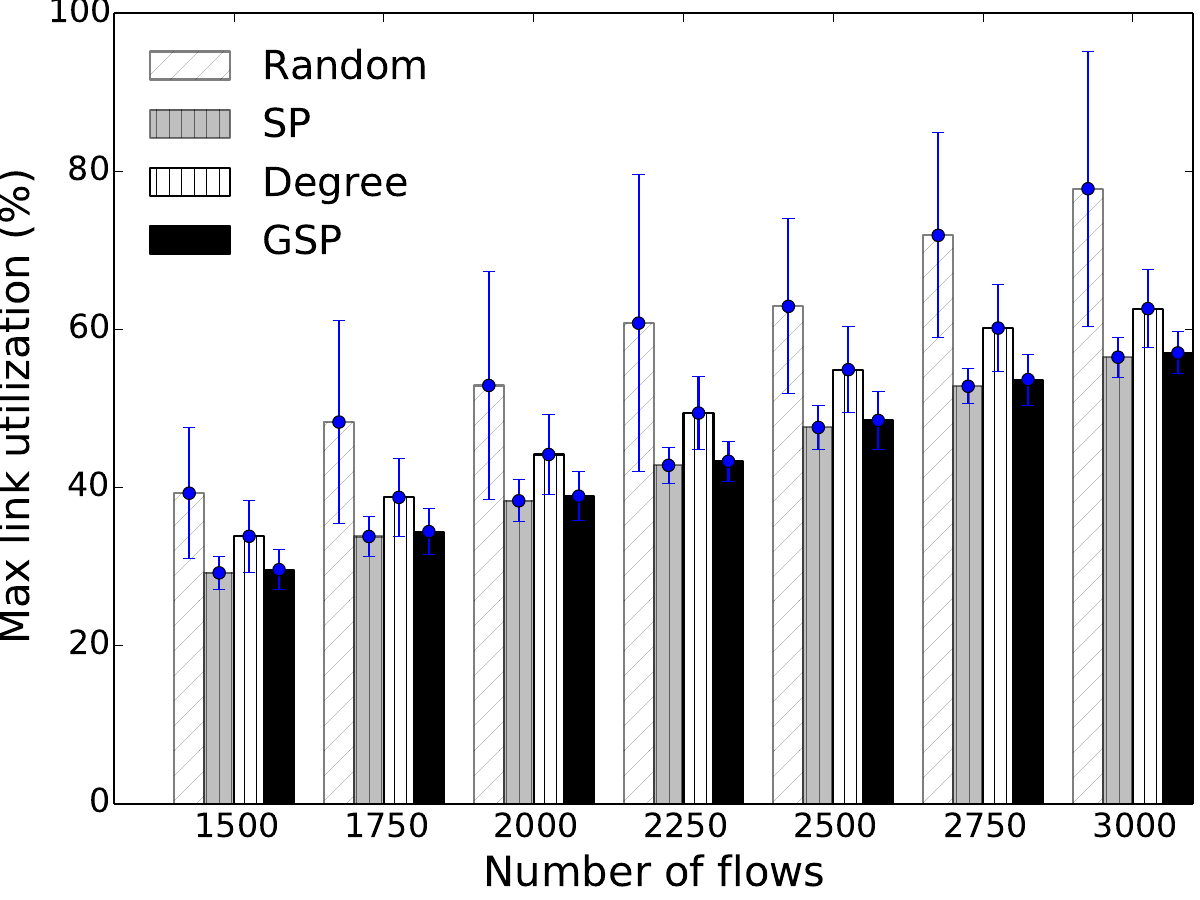}
      \caption{Performance of \telu~on the 100-node network with various centralities. $M$=1 and $K$=4.}
      \label{fig:ul_demands100}
\end{minipage}
\begin{minipage}{0.49\linewidth}
\includegraphics[width=1\linewidth]{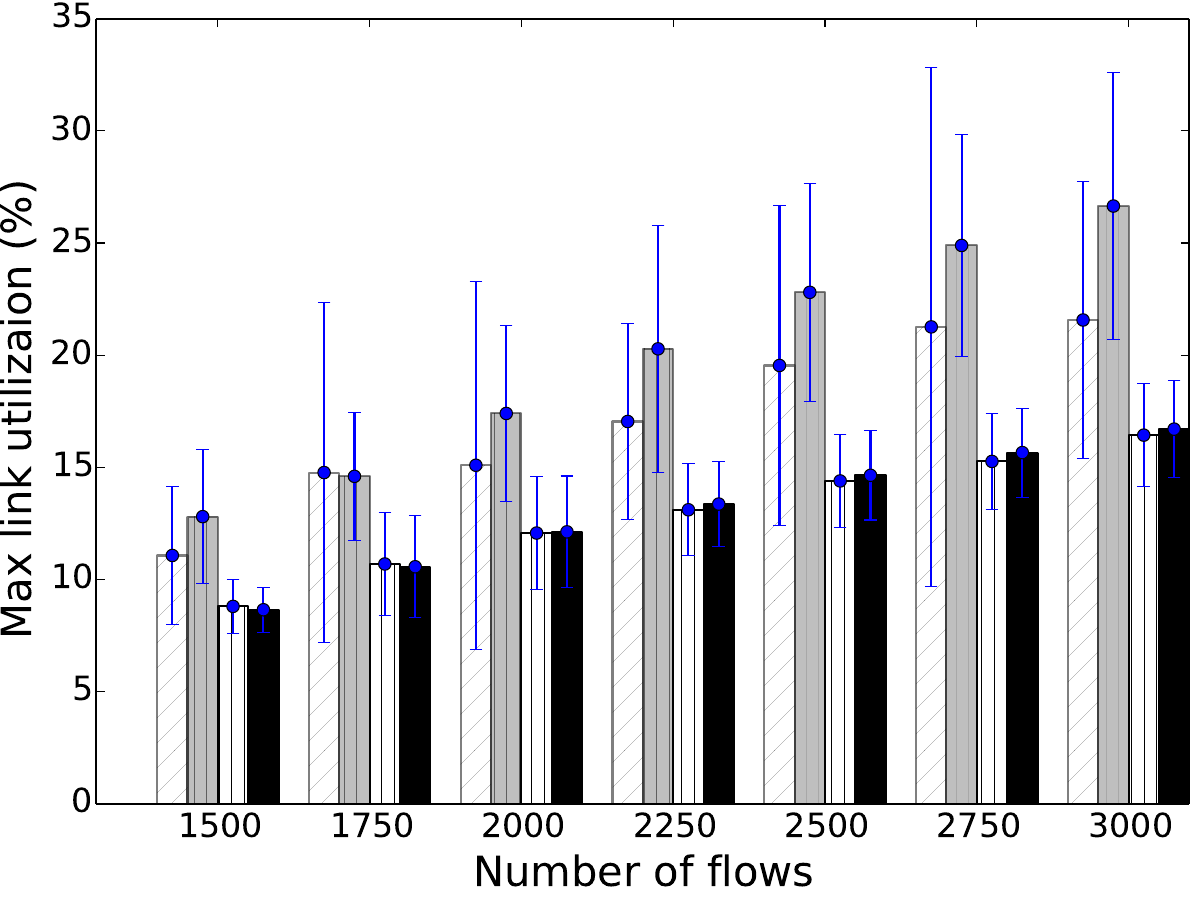}
      \caption{Performance of \telu~on the 161-node network with various centralities. $M$=1 and $K$=4.}
      \label{fig:ul_demands161}
\end{minipage}
% \vspace{-3mm}
\end{figure}

\begin{figure}[htpb]
\centering
\begin{minipage}{0.49\linewidth}
   \includegraphics[width=1\linewidth]{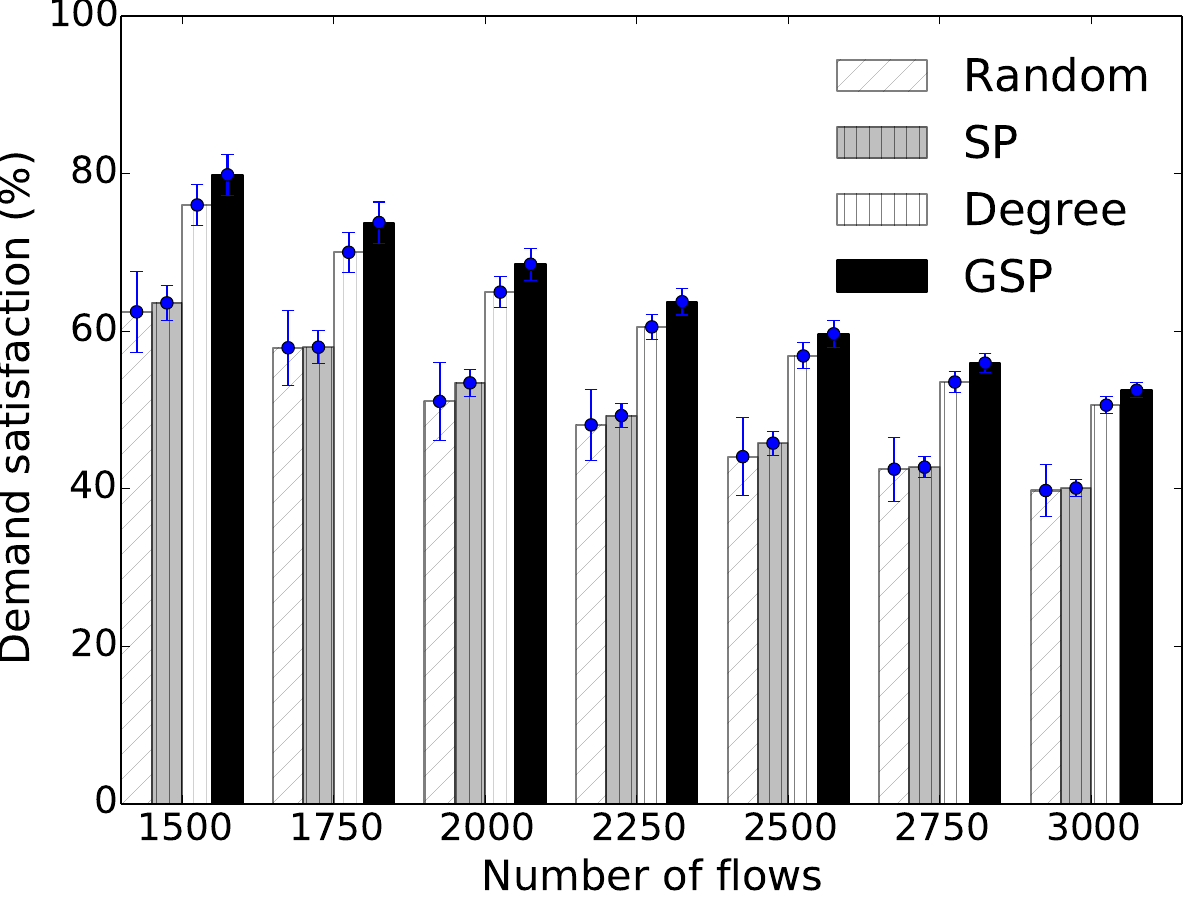}
      \caption{Performance of \temf~on the 100-node network with various centralities. $M$=1 and $K$=7.}
      \label{fig:mf_100}
\end{minipage}
\begin{minipage}{0.49\linewidth}
\includegraphics[width=1\linewidth]{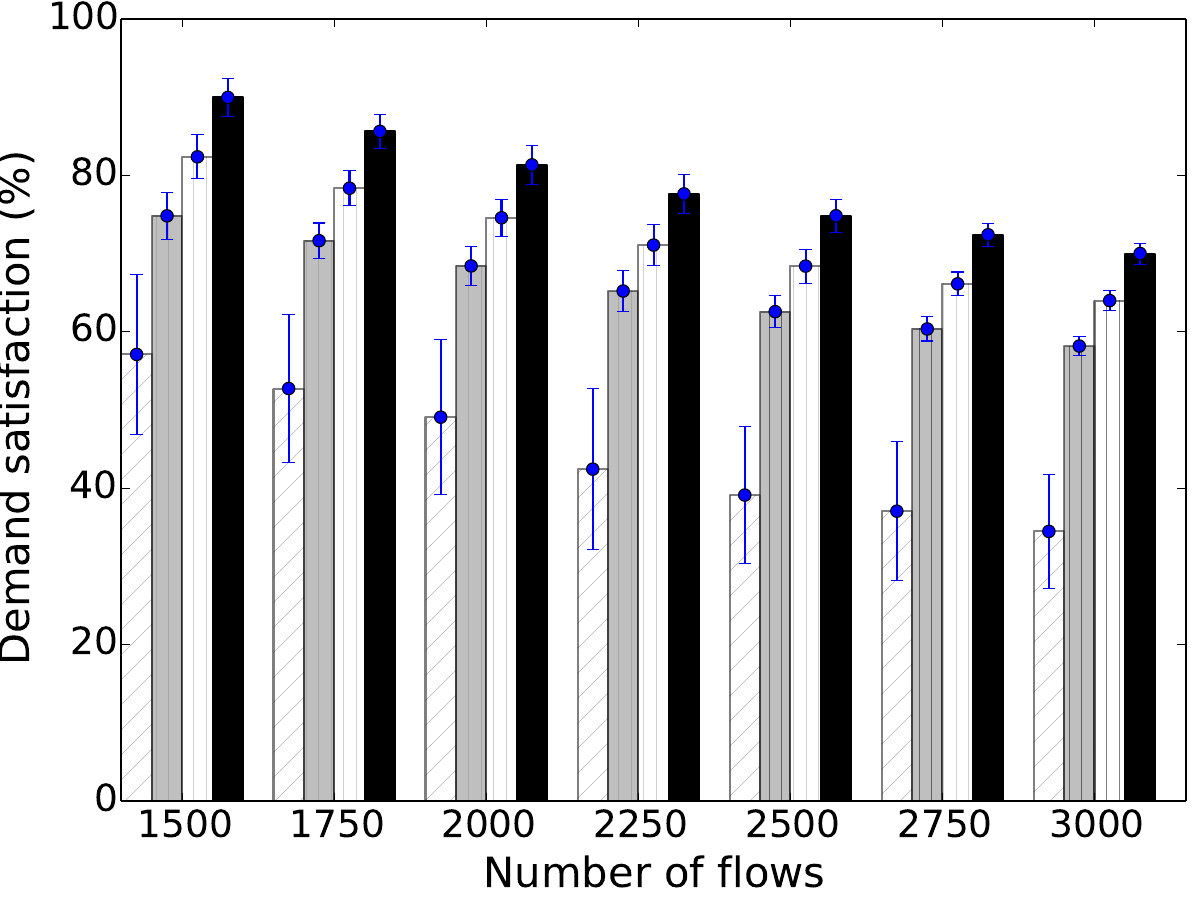}
      \caption{Performance of \temf~on the 161-node network with various centralities. $M$=1 and $K$=7. }
      \label{fig:mf_161}
\end{minipage}
% \vspace{-3mm}
\end{figure}

% \noindent{\bf Centralities without weights. }
Fig.~\ref{fig:ul_demands100} and Fig.~\ref{fig:ul_demands161} depict the results for \telu, and Fig.~\ref{fig:mf_100} and Fig.~\ref{fig:mf_161} for \temf.
We can make several interesting observations.
First, for the 100-node network {\em SP} and {\em GSP} perform the best under all settings in Fig.~\ref{fig:ul_demands100}. In contrast, for the 161-node topology in Fig.~\ref{fig:ul_demands161} {\em GSP} and {\em Degree} perform the best. When considering \temf, Fig.~\ref{fig:mf_100} shows that {\em GSP} and {\em Degree} perform better in the 100-node topology, while in the 161-node topology {\em GSP} performs best in Fig.~\ref{fig:mf_161}.
Thus, middlepoints chosen by group shortest path centrality consistently outperform those selected by other centralities in terms of TE performance.

The main advantage of {\em GSP} is that it selects a set of middlepoints whose combined power is strong. In particular, {\em SP} may select nodes that are individually strong but cover the same set of shortest paths; thus, when combined together these nodes result in poor performance since they share the same shortest paths and are unable to spread out the traffic. This is the reason why {\em GSP} performs consistently well, while the performance of {\em SP} can fluctuate from very strong as in Fig.~\ref{fig:ul_demands100} to very poor and even worse than {\em Random} as in Fig.~\ref{fig:ul_demands161}.

Second, {\em Random} performs the worst in Fig.~\ref{fig:ul_demands100}, Fig.~\ref{fig:mf_100}, and Fig.~\ref{fig:mf_161}, and it also performs badly in the 161-node network in Fig.~\ref{fig:ul_demands161}. This confirms our premise that centrality based middlepoint selection generally outperforms a naive random selection scheme. Indeed, {\em Random} does not utilize any topological information from the network. Further, {\em Random} fluctuates wildly, which makes it ill-fitted for practical use. As seen from the figures, {\em Random} has the largest standard deviations among all.

Third, we observe that the performance of {\em SP} can be worse than {\em Random} sometimes in Fig.~\ref{fig:ul_demands161}. Indeed, {\em SP} just greedily selects the top-$K$ shortest-path central nodes, even though in reality these nodes may share several shortest paths.
{\em Random}, on the other hand, can do better than {\em SP} in certain settings since it has a lower probability of choosing overlapping shortest paths.  

Another aspect of performance is the runtime of the TE LPs.
Table~\ref{table:LP_time_lu_161} and Table~\ref{table:LP_time_mf_161} show the average runtimes for \telu~and \temf~respectively.
 % as well as the total time including both middlepoint selection and solving the LP for the 161-node real network.
{\em Random} consistently has the worst results. %followed by {\em Degree}.
{\em SP} takes the least time but the difference between {\em SP}, {\em GSP}, and {\em Degree} is little.
All of the schemes can finish within 100 seconds even with 2000 flows, which demonstrates that centrality based segment routing can be practically used in large-scale networks.

%\begin{table}[!h]
%\centering
%\resizebox{0.98\columnwidth}{!}{%
%\begin{tabular}{|l|c|c|c|c|c|c|c|}
%\hline
%\backslashbox{Scheme}{\# demands}&500&750&1000&1250&1500&1750&2000 \\ \hline
%% Synthetic&synth50&50&276&2449 \\ \hline
%Random&11.52&17.49&23.78&30.23&37.32&43.59&51.92 \\ \hline
%SP&7.84&12.21&16.54&21.86&27.20&32.86&37.36 \\ \hline
%Degree&9.06&15.59&19.40&25.49&30.43&36.18&42.27 \\ \hline
%GSP&8.12&13.05&16.98&22.60&28.68&33.27&38.87 \\ \hline
%\end{tabular}
%}
%\vspace{3mm}
%\caption{Average LP time (seconds) of 161-node network of \telu.}
%\label{table:LP_time_lu_161}
%\vspace{-3mm}
%\end{table}
%
%
%\begin{table}[!h]
%\centering
%\resizebox{0.98\columnwidth}{!}{%
%\begin{tabular}{|l|c|c|c|c|c|c|c|}
%\hline
%\backslashbox{Scheme}{\# demands}&500&750&1000&1250&1500&1750&2000 \\ \hline
%% Synthetic&synth50&50&276&2449 \\ \hline
%Random&17.99&29.86&42.68&55.26&69.15&81.62&96.97 \\ \hline
%SP&12.87&20.56&29.42&39.44&49.55&58.83&70.53 \\ \hline
%Degree&14.08&23.11&34.41&43.81&53.79&64.79&76.04 \\ \hline
%GSP&14.34&22.56&32.15&43.27&54.37&66.32&79.54 \\ \hline
%\end{tabular}
%}
%\vspace{3mm}
%\caption{Average LP time (seconds) of 161-node network of \temf.}
%\label{table:LP_time_mf_161}
%\vspace{-3mm}
%\end{table}

\begin{table}[!h]
\centering
\resizebox{0.98\columnwidth}{!}{%
\begin{tabular}{|l|c|c|c|c|c|c|c|}
\hline
\backslashbox{Scheme}{\# flows}&1500&1750&2000&2250&2500&2750&3000 \\ \hline
% Synthetic&synth50&50&276&2449 \\ \hline
Random&37.54&44.76&53.07&60.70&68.18&78.52&88.00 \\ \hline
SP&28.16&32.61&38.92&44.42&50.93&56.65&62.14 \\ \hline
Degree&30.81&35.75&41.40&47.68&54.56&61.99&68.01 \\ \hline
GSP&27.48&33.75&38.49&45.44&50.56&57.38&62.64 \\ \hline
\end{tabular}
}
\vspace{3mm}
\caption{Average LP time (seconds) of 161-node network with \telu.}
\label{table:LP_time_lu_161}
\vspace{-4mm}
\end{table}

\begin{table}[!h]
\centering
\resizebox{0.98\columnwidth}{!}{%
\begin{tabular}{|l|c|c|c|c|c|c|c|}
\hline
\backslashbox{Scheme}{\# flows}&1500&1750&2000&2250&2500&2750&3000 \\ \hline
% Synthetic&synth50&50&276&2449 \\ \hline
Random&69.15&81.62&96.97&115.41&137.69&153.45&176.62 \\ \hline
SP&49.55&58.83&70.53&79.87&93.67&107.57&122.04 \\ \hline
Degree&53.79&64.79&76.04&90.49&102.63&112.53&129.58 \\ \hline
GSP&54.37&66.32&79.54&94.04&108.38&125.80&140.17 \\ \hline
\end{tabular}
}
\vspace{4mm}
\caption{Average LP time (seconds) of 161-node network with \temf.}
\label{table:LP_time_mf_161}
\vspace{-3mm}
\end{table}

The reason that {\em Random} has the longest runtime is that it selects nodes that are not central with possibly many distinct paths and links. This leads to more active optimization variables and constraints for the same LP, thus longer runtime. By the same token, the reason that {\em SP} has the lowest runtime is that it selects top-$K$ central nodes with many overlapping shortest paths. This results in fewer active links being used for routing, and thus fewer active optimization variables and constraints in the LP.

To summarize, based on the above experimental results and analysis, we find that {\em GSP} consistently delivers the best TE performance with the least LP time among all centralities we considered.
% time-saving middlepoint selection method to be used with TE.

\subsection{Comparison of Weighted Centralities}
\label{sec:eva_weighted}

The centralities we have studied so far only considered the connectivity of the network topology.
As discuss in \cref{sec:2pass}, it is also possible to take into account the link capacity information by adding weights to links and using weighted versions of centralities.

\begin{figure}[htpb]
\centering
\begin{minipage}{0.49\linewidth}
   \includegraphics[width=1\linewidth]{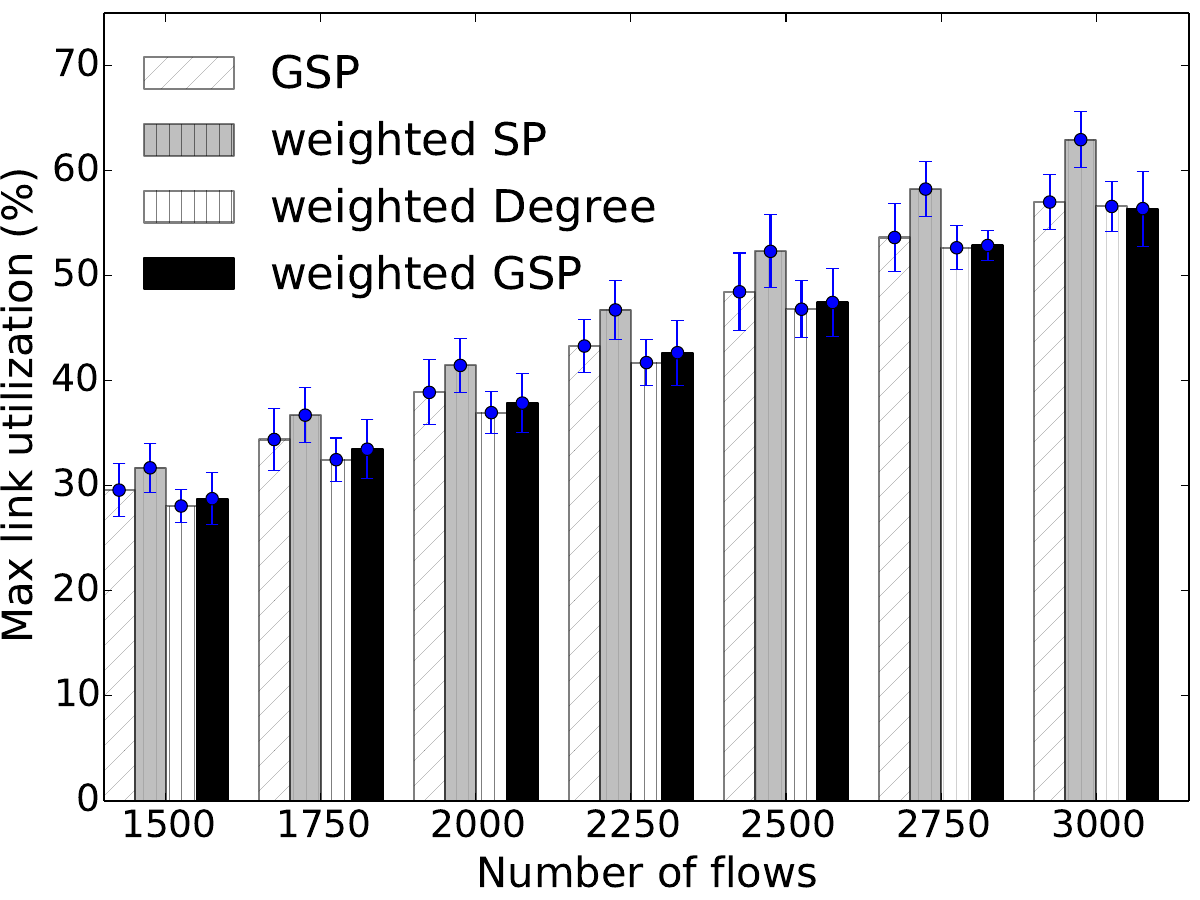}
      \caption{100-node network when $M$=1 and $K$=4 based on weighted centrality for \telu.}
      \label{fig:weighted_ul_100}
\end{minipage}
\begin{minipage}{0.49\linewidth}
\includegraphics[width=1\linewidth]{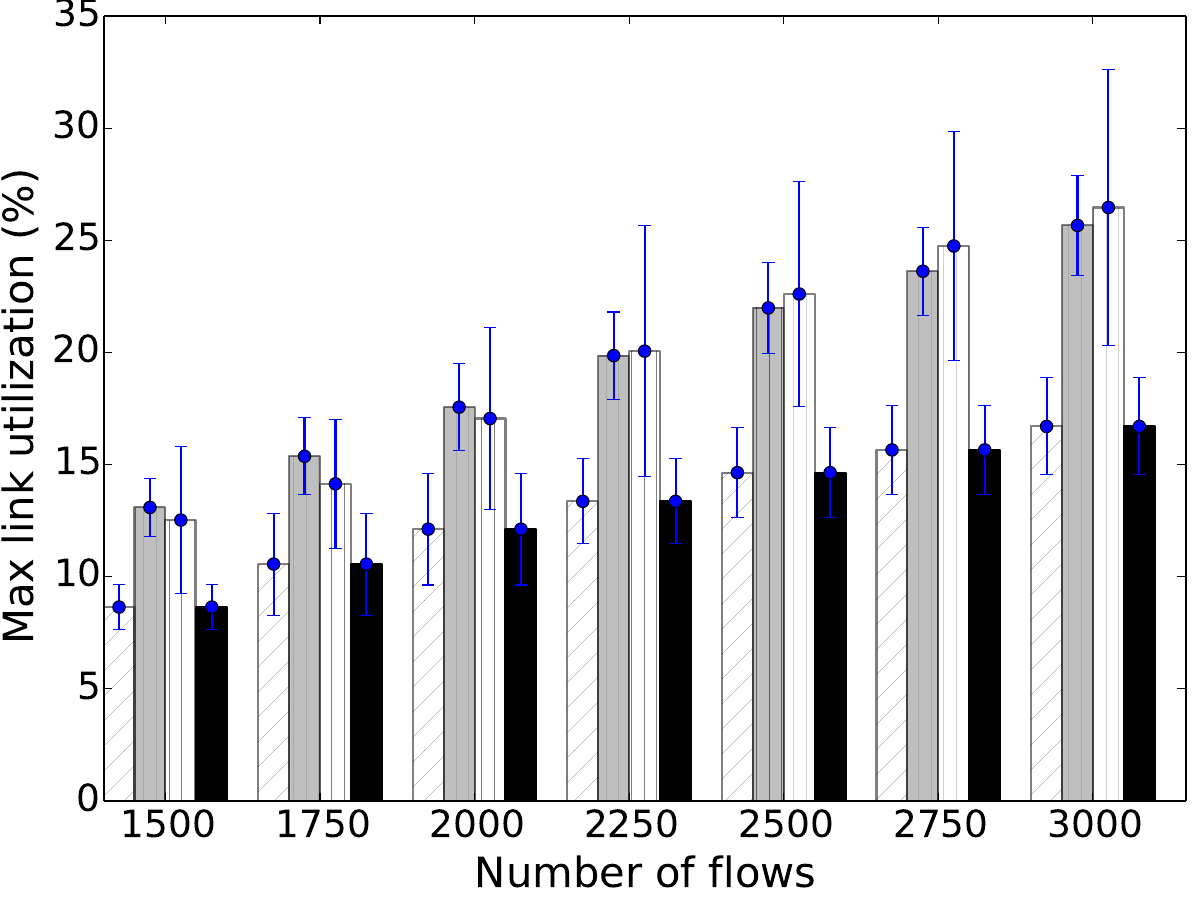}
      \caption{161-node network when $M$=1 and $K$=4 based on weighted centrality for \telu.}
      \label{fig:weighted_ul_161}
\end{minipage}
% \vspace{-3mm}
\end{figure}

\begin{figure}[htpb]
\centering
\begin{minipage}{0.49\linewidth}
   \includegraphics[width=1\linewidth]{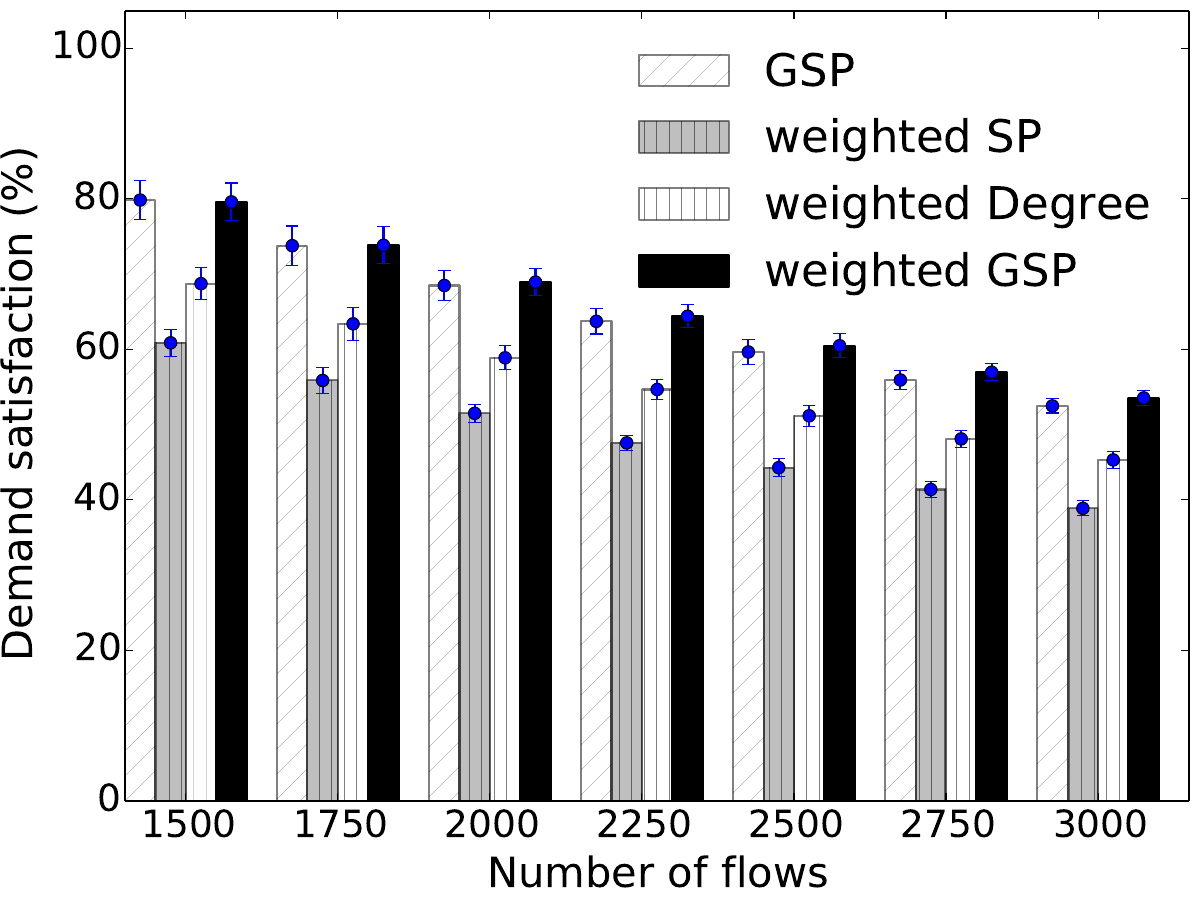}
      \caption{100-node network when $M$=1 and $K$=7 based on weighted centrality for \temf.}
      \label{fig:weighted_mf_100}
\end{minipage}
\begin{minipage}{0.49\linewidth}
\includegraphics[width=1\linewidth]{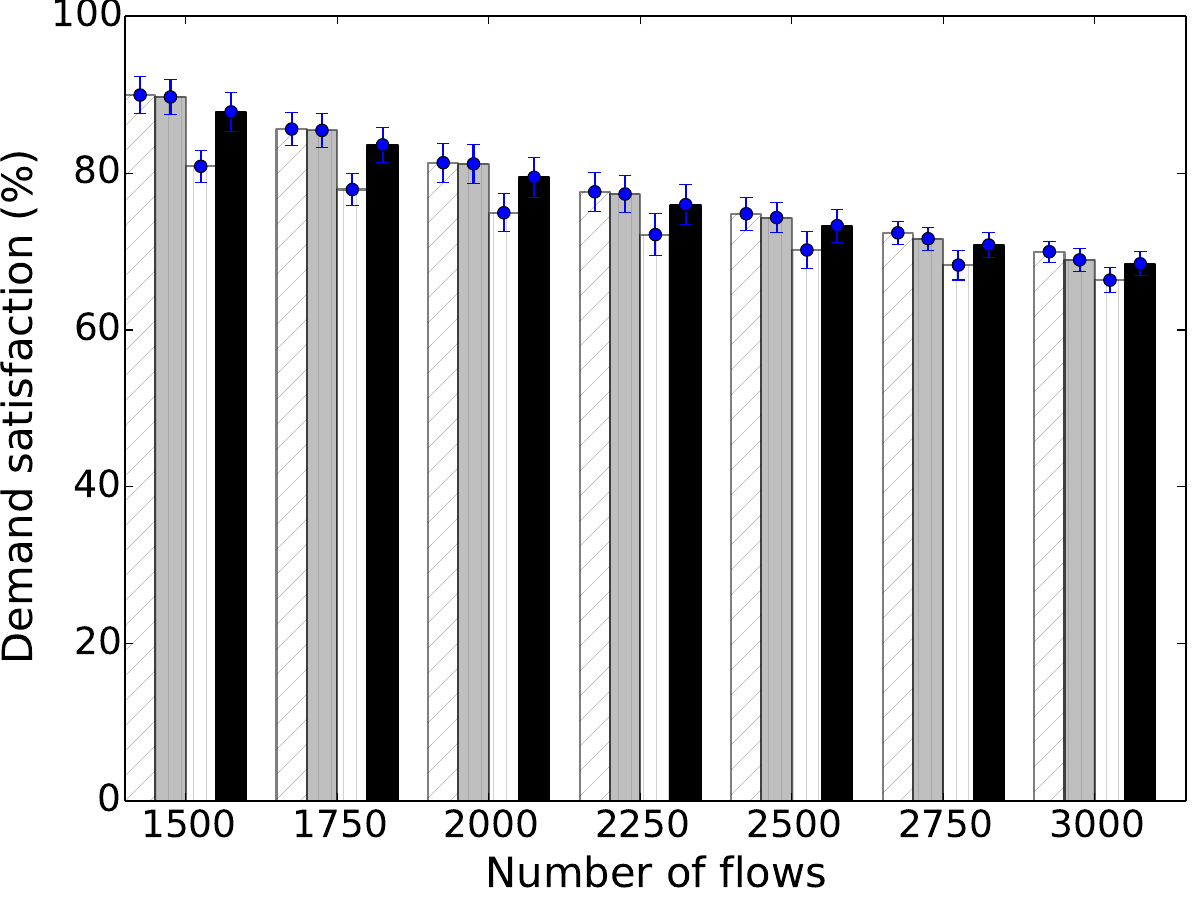}
      \caption{161-node network when $M$=1 and $K$=7 based on weighted centrality for \temf.}
      \label{fig:weighted_mf_161}
\end{minipage}
% \vspace{-3mm}
\end{figure}

We also carry out experiments to compare the performance of weighted {\em SP}, weighted {\em degree}, and weighted {\em GSP} centralities against {\em GSP}, the best centrality without using weights for middlepoint selection. Weighted here means that the three centrality based approaches are weighted by the capacity of each edge.
Fig.~\ref{fig:weighted_ul_100} and Fig.~\ref{fig:weighted_ul_161} show the performance comparison with \telu, while Fig.~\ref{fig:weighted_mf_100} and Fig.~\ref{fig:weighted_mf_161} show the comparison with \temf. For the 161-node topology, we observe that {\em GSP} and weighted {\em GSP} are always the best. In the 100-node topology, {\em GSP} is sometimes worse than weighted {\em Degree} and weighted {\em GSP} although the differences are very small. Therefore, {\em GSP} without weights is still the most effective and robust middlepoint selection method in all settings. 
% Using link capacities does not improve the TE performance significantly. 

%!TEX root = main.tex
\section{Related Work}
\label{sec:related}

We now review related work on segment routing other than those discussed already in \cref{sec:bg}. Segment routing is a relatively new concept with limited prior work. 
Aubry et al. \cite{aubry2016scmon} propose to use segment routing for continuous monitoring of the data plane of the network with a single box. Segment routing is used to force probe packets to traverse specific paths. Giorgetti et al. \cite{giorgetti2015path} propose algorithms for segment routing label stack computation that guarantee minimum label stack depth.

TE has been extensively studied in carrier networks \cite{EJLW01,KKDC05,WXQY06,FT00,HBCR07,HVSB15}, and has also attracted much attention recently in data center backbone WANs \cite{JKMO13,HKMZ13,LKMZ14,GHCR13} with software defined networking \cite{FRZ13}. End-to-end paths are usually used while we study segment routing here in TE. 

Finally, we note that graph centralities have been applied to routing in some specific SDN problems, such as in service chain embedding \cite{Lukovszki2015} and incremental SDN deployment \cite{Lukovszki2016,Levin2014}. In a service chain \cite{Lukovszki2015}, traffic needs to be steered through a set of waypoints, with the goal of admitting a maximum number of routes. In the context of hybrid and incremental SDN deployment \cite{Levin2014}, a set of middleboxes need to be deployed in order to serve a maximal number of flows, respecting flow rule constraints.
Solutions to these problems are based on degree centralities, and there exist greedy approximation algorithms exploiting submodularity as well \cite{Lukovszki2016}. Contrary to these works, our paper focuses on the theoretical fundamentals of TE using segment routing and on graph-theoretic practical middlepoint selection.
\section{Conclusion}
\label{sec:conclusion}

In this work, we studied practical traffic engineering with segment routing in SDN based WANs. We showed that TE for segment routing with shortest paths is (weakly) polynomial when the number of middlepoints per logical path is fixed and not part of the input. We also studied practical TE with shortest path based segment routing, and proposed to select just a few important nodes for all network traffic using graph theoretic centrality concepts. Our performance evaluation demonstrated that just a small percentage of powerful nodes can achieve good results at very low time complexities.

%We have proposed MidwayTE, a novel TE architecture with segment routing in SDN based WANs, which selects only few central switches as middlepoints to route the entire network traffic. To define node centrality, we employ several concepts from graph theory. Furthermore, we theoretically investigate flow centralities and acyclic segment routing, and provide several hardness results. Our experimental evaluation with realistic topologies and demands demonstrates that just a small percentage of powerful nodes can vastly simplify network management by orders of magnitude, while achieving low running times and satisfactory link utilizations.

%\color{black}

%\section*{Acknowledgment}
%This research is supported by Hong Kong RGC grant HKU 716712E.



%\begin{spacing}{0.9}
\bibliographystyle{IEEEtranS}
\bibliography{main}

\end{document}